\documentclass[10pt,a4paper]{amsart}

\usepackage[a4paper,margin=1.5cm]{geometry}
\usepackage{amsmath}
\usepackage{amsfonts}
\usepackage{amssymb}
\usepackage{graphicx}
\usepackage{booktabs}
\usepackage{tabularx}
\usepackage{ragged2e}
\usepackage{caption}
\usepackage{microtype}
\allowdisplaybreaks
\usepackage{hyperref}
\hypersetup{hidelinks}
\usepackage{cite}
\usepackage{xcolor}
\usepackage{placeins}
\usepackage{float}
\usepackage{subcaption}
\usepackage{array}
\usepackage{multicol}

\graphicspath{{figures/}}

\newtheorem{thm}{Theorem}[section]
\newtheorem{cor}[thm]{Corollary}
\newtheorem{lem}[thm]{Lemma}
\newtheorem{prop}[thm]{Proposition}
\theoremstyle{definition}
\newtheorem{defn}[thm]{Definition}

\theoremstyle{remark}
\newtheorem{rem}[thm]{Remark}

\numberwithin{equation}{section}

\begin{document}


\title[Evolutionary Entropy Shapes Reproductive Lifespan]
{Evolutionary Entropy Shapes Reproductive Lifespan \\ in Age-Structured Populations}

\author[J. Buescu]{Jorge Buescu}
\thanks{Jorge Buescu, Departamento de Matemática, Faculdade de Ciências da Universidade de Lisboa, Centro de Estudos Matemáticos, Universidade de Lisboa, Portugal}
\email{jsbuescu@ciencias.ulisboa.pt}

\author[S. N. Elaydi]{Saber N. Elaydi}
\thanks{Saber N. Elaydi, Department of Mathematics, Trinity University, San Antonio, Texas, USA}
\email{selaydi@trinity.edu}

\author[H. M. Oliveira]{Henrique M. Oliveira}
\thanks{Henrique M. Oliveira, Corresponding author, Departamento de Matemática, Instituto Superior Técnico, Universidade de Lisboa, Centro de Análise Matemática, Geometria e Sistemas Dinâmicos, Universidade de Lisboa, Portugal}
\email{henrique.m.oliveira@tecnico.ulisboa.pt}
\thanks{The third author acknowledges partial support by Fundação para a Ciência e a Tecnologia, UIDB/04459/2020 and UIDP/04459/2020}




\date{}


\begin{abstract}
Evolutionary entropy measures the temporal organization of reproductive contributions along the life cycle of an age-structured population. We develop a mathematical and empirical framework showing that, in iteroparous animal populations represented by Leslie-type demographic matrices, reproductive windows are frequently organized near the age classes selected by entropy maximization.

Evolutionary entropy is placed explicitly alongside the classical demographic net reproductive number, denoted here by $R_0^{\rm dem}$ to distinguish it from the epidemiological basic reproduction number. Whereas $R_0^{\rm dem}$ measures total lifetime reproductive replacement and $\lambda$ measures asymptotic population growth, evolutionary entropy measures the temporal dispersion of the growth-adjusted reproductive distribution. This distinction gives the paper a biological interpretation: replacement, growth, and reproductive organization are complementary demographic quantities rather than interchangeable ones.

The central mathematical result is a reduction principle: under the Euler--Lotka normalization, evolutionary entropy and generation time are invariant under multiplicative rescaling of survivorship and fertility on the reproductive interval. Thus the entropy relevant to reproductive timing is not determined by the absolute level of survivorship or fertility, nor by juvenile mortality alone, but by the normalized post-maturity reproductive distribution. This distinction separates the physical survivorship curve from the probability distribution of realized reproductive contributions.

We derive explicit entropy functionals for finite Leslie models and open-group Leslie models, including finite and infinite geometric reproductive tails. For the geometric reproductive regime, governed by the effective ratio $q=\rho/\lambda$, we prove that a sharp critical threshold separates populations with a unique finite entropy-maximizing reproductive endpoint from populations whose entropy increases toward an asymptotic value. The threshold is determined by the unique solution of $q^A+q-1=0$, where $A$ is age at first reproduction.

The theory is tested on 130 animal species represented by Leslie-type demographic matrices. Entropy-derived reproductive endpoints, reproductive-distribution quantiles, and threshold predictions are compared with independent life-history variables, including age at first reproduction, mean reproductive age, reproductive lifespan, life expectancy, and maximum longevity. The predictions are computed from the demographic matrices alone, before any external life-history variables are introduced.

The empirical agreement is strong. Predicted and observed reproductive medians coincide exactly for a majority of species, more than $90\%$ of species are predicted within three reproductive classes, and the main associations remain strong after phylogenetic correction based on the Open Tree of Life. Distributional comparisons using overlap and Bhattacharyya coefficients further show that the entropy-derived windows capture the shape of reproductive allocation, not merely a single summary age.

These results identify a quantitative regularity across diverse taxa: the distribution of reproductive contributions encodes biologically meaningful information about reproductive timing. They further suggest that geometric reproductive distributions, which arise naturally in open-group Leslie populations and under approximately constant adult hazards, play a central role in the entropy-based description of reproductive organization.
\end{abstract}

\maketitle


\section{Introduction}

\begin{figure*}[ht]
\centering
\includegraphics[width=0.85\textwidth]{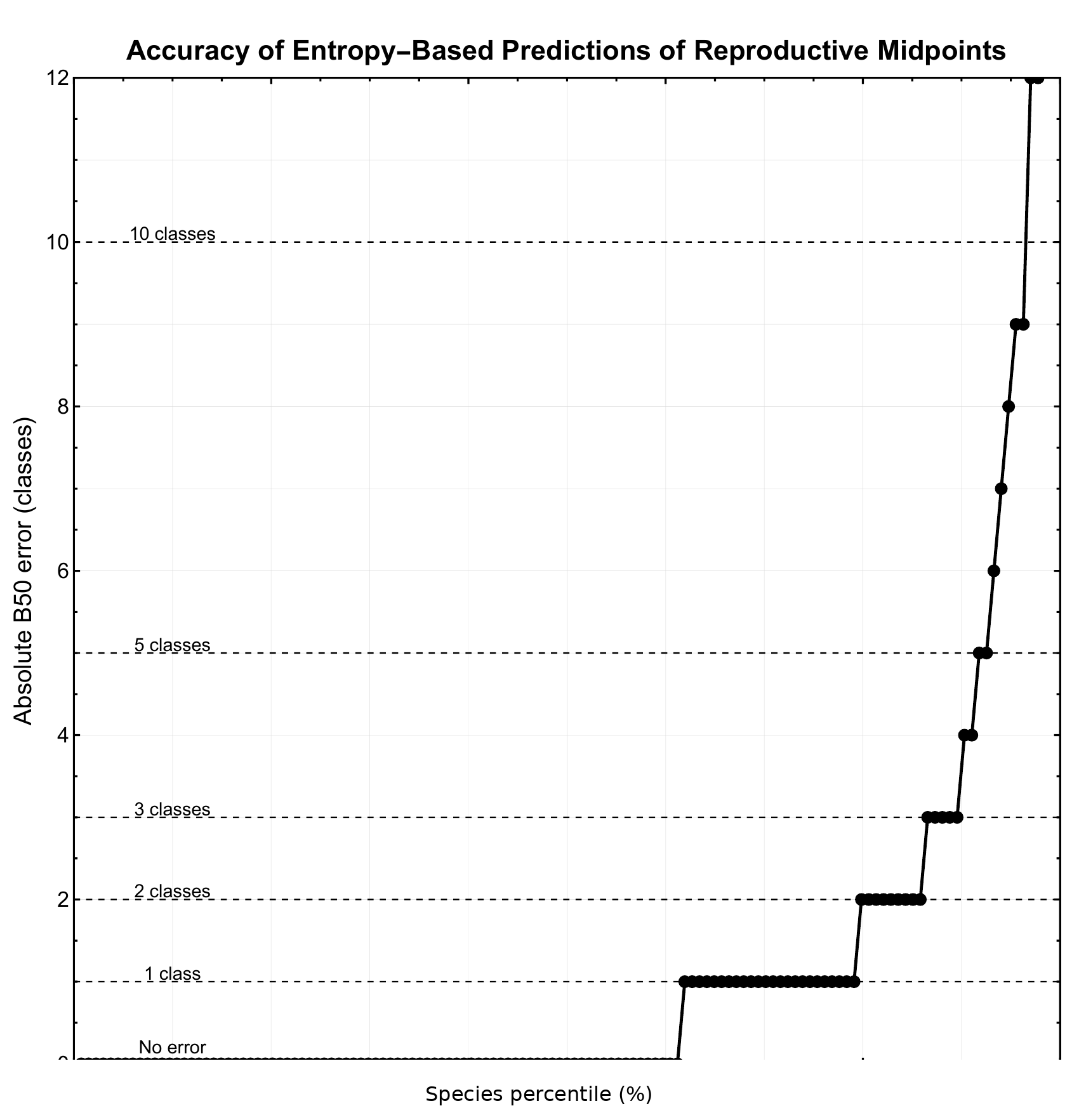}
\caption{Absolute prediction errors of the reproductive median $B_{50}$, ordered by species percentile. Horizontal dashed lines indicate error thresholds of 1, 2, 3, 5 and 10 reproductive classes.}
\label{fig:B50PercentileErrors}
\end{figure*}

What determines the duration of reproductive life remains one of the
central problems of evolutionary demography and life-history theory
\cite{Medawar1952,Williams1957,Ham,Kirkwood1977,
Stearns1992,Charlesworth,Kirkwood2000,Baudisch2008}.
Species exhibit strikingly different survivorship organizations,
ranging from approximately rectangular late-life survival to strongly
declining demographic structures
\cite{Pearl1928,Deevey,Demetrius78,Jones2014}.
Yet reproductive lifespan frequently occupies comparatively narrow
scales relative to total longevity
\cite{Carey2003,Fin,deMagalhaes2009,Jones2009PanTHERIA,PanTHERIA,
SalgueroGomez2016COMADRE,COMADRE}.
Understanding how reproductive windows are organized across the life
cycle therefore remains a fundamental challenge in evolutionary
biology.

Age-structured population theory provides a natural foundation for
investigating this problem
\cite{Leslie1945,Keyfitz1968,Cus,ElaydiCushing2024,Charlesworth}.

A first biological descriptor of such a population is the demographic net reproductive number. In a life table or Leslie model this quantity, denoted here by $R_0^{\rm dem}$, is the expected total reproductive output of a newborn individual over its lifetime after survivorship to each reproductive age has been taken into account. It is the classical replacement index of demography: values above one indicate reproductive excess, values equal to one indicate replacement, and values below one indicate reproductive deficit. In this sense $R_0^{\rm dem}$ answers the most basic demographic question: does the lifetime schedule of survival and fertility replace the individuals who generated it?

However, reproductive replacement does not by itself describe the temporal organization of reproduction. Two species may have comparable net reproductive numbers but very different life histories if one concentrates reproduction shortly after maturity and another spreads the same total reproductive output over many later age classes. The first schedule may generate faster population growth because descendants appear earlier and themselves begin reproducing earlier, whereas the second schedule may have the same lifetime replacement but a different generation time and a different distribution of reproductive contribution. Thus, beyond the amount of lifetime reproduction, one must also measure the timing and dispersion of reproduction across the life cycle.

Within this context, Demetrius introduced evolutionary entropy as a
measure of the temporal reproductive distribution
along genealogical lineages
\cite{Demetrius78,Demetrius83,Demetrius97,Demetrius2004,
Demetrius2013}.
Evolutionary entropy links survivorship, reproduction and generation
time, and has become one of the principal mathematical concepts
connecting demography, ageing and life-history evolution
\cite{Demetrius83,Demetrius97,Demetrius2013}.

Despite its extensive theoretical development, a fundamental biological
question remains largely unexplored. Which demographic structures are
actually responsible for the behaviour of evolutionary entropy?
Does entropy depend primarily on the complete survivorship curve, or on
the structure of reproduction after maturity?

The problem is both mathematical and biological. Mathematically, one must identify the functional dependence of evolutionary entropy on the Euler--Lotka weights and determine how the entropy changes as the reproductive endpoint is varied. Biologically, one must decide whether the resulting entropy-optimal age classes correspond to observable reproductive windows rather than merely to formal extrema of a demographic functional. The analysis below addresses these issues by reducing the entropy calculation to the normalized reproductive probability distribution, deriving the corresponding entropy functionals for finite and open-group Leslie populations, establishing the geometric-tail threshold separating finite entropy maximization from asymptotic saturation, and testing the resulting predictions against independent life-history information across a broad comparative data set.

The present article shows that evolutionary entropy is fundamentally a
property of reproductive organization rather than survivorship itself.
Under constant growth, entropy is invariant under multiplicative
rescaling of survivorship and fertility throughout the reproductive
interval. Consequently, the relevant demographic object is the
reproductive distribution after maturity rather than
the complete life-history survivorship curve.

This observation has important biological consequences.
Classical Pearl--Deevey survivorship curves incorporate juvenile
mortality, developmental filtering and pre-reproductive demographic
effects \cite{Pearl1928,Deevey}.
Survivorship profiles that appear very different over the complete life
cycle may therefore generate remarkably similar reproductive
distributions.

A second result concerns the duration of reproductive life.
We show that evolutionary entropy may exhibit three distinct regimes:
finite entropy-maximizing reproductive age classes, boundary optima,
and asymptotic convergence.
Particular importance is attached to geometric reproductive
distributions, which arise naturally in open-group Leslie populations
and under approximately constant adult hazards.
These distributions admit a complete analytical characterization
through a critical threshold separating finite reproductive optima from
asymptotic entropy regimes.

The empirical motivation for this analysis is reinforced by the growing
availability of comparative demographic databases spanning vertebrates,
mammals, plants and matrix population models
\cite{Tacutu2018,deMagalhaes2009,
SalgueroGomez2015COMPADRE,AnAgeDatabase,PanTHERIA,
COMPADRE,COMADRE},
together with increasing evidence that ageing trajectories exhibit
substantial diversity across the tree of life
\cite{Jones2014,VAU}.

Using 130 animal species represented by Leslie-type demographic
matrices, we compare entropy-derived reproductive classes,
reproductive distribution quantiles and critical-threshold predictions
with independent life-history information, including age at first
reproduction, reproductive lifespan, life expectancy and maximum
longevity.

To account for shared evolutionary history, the principal
life-history comparisons are further repeated using
phylogenetically corrected analyses based on the Open Tree of Life \cite{OpenTreeOfLife}.
Because species cannot generally be regarded as statistically
independent data points owing to common ancestry
\cite{Felsenstein1985,HarveyPagel1991,Pagel1999},
phylogenetic correction provides a more stringent test of the
observed associations. As shown in
subsection~\ref{subsec_phylogenetic}, the resulting associations
remain strong after phylogenetic correction, indicating that the
observed patterns cannot be explained solely by taxonomic
relatedness.

The agreement between theory and observation is striking.
Figure~\ref{fig:B50PercentileErrors} illustrates one of the
central empirical findings of this study: in 85.8\% of species,
entropy-derived predictions of reproductive midpoints differ from the
observed values by at most one or two reproductive classes.

Since open-group matrices represent approximately 74\% of the Leslie
populations retained in our dataset, and naturally generate geometric
reproductive distributions, geometric reproductive organization
appears not as an exceptional demographic structure but as a generic
feature of iteroparous animal populations.

The complete species-level results for all 130 species analysed in the
study are reported in Appendix~\ref{app:species_results}. These tables include the
evolutionary entropy estimates, entropy-maximizing reproductive age
classes, reproductive distribution quantiles, critical-threshold
predictions, observed life-history variables, and the corresponding
comparisons between theoretical predictions and independent empirical
observations.

The results are consistent with the interpretation that
entropy-maximizing age classes provide biologically meaningful
lower bounds for reproductive lifespan. Species may reproduce
beyond these classes and, within the present dataset,
reproductive windows seldom terminate before the
entropy-relevant regime is attained.

The extraordinary diversity of ageing and survivorship
patterns observed across the tree of life
\cite{BuescuOliveiraSousa2023,Jones2014,Baudisch2008,deVries2023}
may conceal a surprisingly simple regularity:
reproductive life is often organized around
entropy-maximizing distributions
\cite{Demetrius78,Demetrius2004}.

The article is organized as follows.
Section~\ref{sec_Leslie} introduces  demographic theory and
evolutionary entropy.
Section~\ref{sec_exp_model} develops the structural properties of the
entropy functional.
Section~\ref{sec_entropy_maximization} establishes the main
theoretical results.
Section~\ref{sec_Validation} presents the empirical validation,
including the phylogenetically corrected analyses of
subsection~\ref{subsec_phylogenetic}.
Appendix~\ref{app:species_results} contains the complete species-level results.


\section{Leslie framework and notation}
\label{sec_Leslie}
\subsection{Pure Leslie}

We consider discrete age-structured populations governed by the
classical Leslie formalism. For the convenience of readers less
familiar with the subject, we briefly review the main theoretical elements; see
\cite{Leslie1945,Keyfitz1968,Cus,ElaydiCushing2024,Charlesworth}
for detailed expositions.

In this setting, individuals are distributed into age classes, and
population dynamics is determined jointly by age-dependent fertility
and survivorship.
The long-term demographic behaviour is governed by the Perron
dominant eigenvalue of the Leslie matrix,
whereas the reproductive organization of the population is encoded
in the Euler--Lotka equation.

Age classes are indexed by
$
j=1,\dots,\omega,
$
where \(\omega\) denotes the maximal lifespan.
Let
$
n_j(k)\in\mathbb N
$
denote the number of individuals in age class \(j\) at discrete time
\(k\).
The population vector is given by
\[
N_k=
\begin{bmatrix}
n_1(k)\\
n_2(k)\\
\vdots\\
n_\omega(k)
\end{bmatrix}
\]
and the demographic evolution is governed by the Leslie equation
\[
N_{k+1}=LN_k,
\]
where
\[
L=
\begin{bmatrix}
m_1 & m_2 & \cdots & m_{\omega-1} & m_\omega\\
s_1 & 0 & \cdots & 0 & 0\\
0 & s_2 & \cdots & 0 & 0\\
\vdots & \vdots & \ddots & \vdots & \vdots\\
0 & 0 & \cdots & s_{\omega-1} & 0
\end{bmatrix}.
\]

Here \(m_j\ge0\) denotes the fertility of age class \(j\), whereas
\(s_j\in[0,1]\) denotes the transition survivorship from age class
\(j\) to age class \(j+1\).
The cumulative survivorship from birth to age class \(j\) is therefore given by 
\[
l_j=
\prod_{k=1}^{j-1}s_k,
\]
where the convention \(l_1=1\) is used.
The corresponding Euler--Lotka (EL) equation is
\begin{equation}
	\label{eq:Euler-Lotka}
	\sum_{j=1}^{\omega}
	\frac{l_jm_j}{\lambda^j}
	=
	1,
\end{equation}
where \(\lambda\) denotes the Perron dominant eigenvalue of the
Leslie matrix. 

\begin{defn}[EL survivorship weights] We define the corresponding {\em EL survivorship weights} $\ell_j$ by 
\begin{equation}
	\label{eq:EL_survivorship}
	\ell_j
	=
	\frac{l_j}{\lambda^j}.
\end{equation}
\end{defn}

In terms of these quantities, the Euler--Lotka equation \eqref{eq:Euler-Lotka} assumes the
form
\begin{equation}
\label{eq_Euler-Lotka}
\sum_{j=1}^{\omega}\ell_jm_j=1.
\end{equation}

\subsection{The demographic net reproductive number}
\label{subsec_R0_dem}

Before introducing evolutionary entropy, it is useful to recall the classical demographic threshold quantity associated with an age-structured population. For a Leslie population with survivorship sequence $l_j$ and fertility schedule $m_j$, the \emph{demographic net reproductive number} is
\begin{equation}
\label{eq:R0dem}
R_0^{\rm dem}=\sum_{j=1}^{\omega}l_jm_j.
\end{equation}
We use the notation $R_0^{\rm dem}$ in order to distinguish this quantity from the epidemiological basic reproduction number. In demography, $R_0^{\rm dem}$ is the expected number of offspring produced by a newborn individual over its lifetime, after discounting fertility at each age by the probability of surviving to that age. In female-based life tables it is often interpreted as the expected number of daughters produced by a newborn female. More generally, it is the lifetime replacement output generated by the survivorship and fertility schedules \cite{Keyfitz1968,Caswell,Charlesworth,Cus}.

The biological threshold meaning of $R_0^{\rm dem}$ is immediate. If $R_0^{\rm dem}>1$, then a newborn individual is expected, over its lifetime, to produce more than one replacing offspring, and the population has reproductive excess. If $R_0^{\rm dem}=1$, the population is at replacement. If $R_0^{\rm dem}<1$, lifetime reproduction is insufficient for replacement, and the population has reproductive deficit. Thus $R_0^{\rm dem}$ measures whether the lifetime survival-fertility schedule is capable of replacing the individuals who generated it.

The Euler--Lotka equation refines this replacement criterion by incorporating the timing of reproduction. Define
\[
F(\lambda)=\sum_{j=1}^{\omega}\frac{l_jm_j}{\lambda^j}.
\]
Under the usual assumption that at least one reproductive age class has positive reproductive output, $F$ is strictly decreasing for $\lambda>0$. Since $F(1)=R_0^{\rm dem}$ and the Perron growth factor is characterized by $F(\lambda)=1$, we obtain the threshold correspondence
\begin{equation}
\label{eq:R0_lambda_threshold}
R_0^{\rm dem}>1 \Longleftrightarrow \lambda>1,
\qquad
R_0^{\rm dem}=1 \Longleftrightarrow \lambda=1,
\qquad
R_0^{\rm dem}<1 \Longleftrightarrow \lambda<1.
\end{equation}
Consequently, $R_0^{\rm dem}$ determines whether lifetime reproduction is above, equal to, or below replacement, whereas $\lambda$ gives the asymptotic population growth factor per time step.

This distinction is important biologically. The number $R_0^{\rm dem}$ measures the amount of lifetime reproduction, but it does not determine how that reproduction is distributed through time. Early reproduction and delayed reproduction can yield the same value of $R_0^{\rm dem}$ while producing different growth factors, different generation times, and different age distributions of reproductive contribution. Reproduction occurring early in life has greater demographic leverage because offspring enter the population sooner and can themselves contribute sooner. Delayed reproduction may preserve lifetime replacement but change the pace and temporal structure of population renewal.

Evolutionary entropy enters precisely at this point. The quantities
\begin{equation}
\label{eq:growth_adjusted_reproductive_contributions}
p_j=\frac{l_jm_j}{\lambda^j}
\end{equation}
form the Euler--Lotka reproductive distribution because their sum is one by \eqref{eq:Euler-Lotka}. They describe how the normalized, growth-adjusted reproductive contribution is distributed among age classes. Evolutionary entropy is a measure of the temporal dispersion of this distribution. Thus the present paper treats three complementary demographic objects:
\[
R_0^{\rm dem}
\quad\hbox{measures lifetime reproductive replacement,}
\]
\[
\lambda
\quad\hbox{measures asymptotic population growth,}
\]
\[
H
\quad\hbox{measures the temporal organization of reproductive contribution.}
\]
The aim is therefore not merely to decide whether a population replaces itself, but to understand how the reproductive contribution is distributed across the reproductive lifespan once the Euler--Lotka normalization has been imposed.

We now characterize the principal reproductive distributions
that will be used throughout the remainder of the paper.

\medskip

\noindent
\textbf{Reproductive interval.}
Consider a reproductive interval
\[
A\le j\le D,
\qquad
1\le A<D\le\omega,
\]
where \(A\) denotes the age class at reproductive maturity
and \(D\) the final reproductive age class. We assume that
fertility is restricted to this interval, namely
\[
m_j=0,
\qquad
j<A
\quad\text{or}\quad
j>D.
\]

\begin{defn}[Reproductive distribution and reproductive masses]
\label{repdist}
The {\em reproductive distribution} is given by
\[
p_j=\ell_jm_j.
\]
From \eqref{eq_Euler-Lotka},
\[
\sum_{j=1}^{\omega}p_j
=
\sum_{j=A}^{D}p_j
=
1,
\]
and therefore \(\{p_j\}_{j=A}^{D}\) is a probability
distribution and the quantities \(p_j\) are the
corresponding {\em reproductive masses}.
\end{defn}

\medskip

\noindent
\textbf{Constant fertility.}
In some parts of this article we assume that fertility is
constant throughout the reproductive interval. Thus,
for some constant \(f>0\),
\begin{equation}
\label{eq_fertility}
m_j=f,
\qquad
A\le j\le D.
\end{equation}

Under this hypothesis, condition \eqref{eq_Euler-Lotka}
becomes
\[
f
\sum_{j=A}^{D}
\ell_j
=
1.
\]

Define the {\em total survivorship mass} by
\begin{equation}
\label{eq_mass}
S
=
\sum_{j=A}^{D}
\ell_j.
\end{equation}
Then
\[
S=\frac1f.
\]

On the reproductive interval we have the reproductive
masses
\begin{equation}
\label{eq_p_j}
p_j
=
\frac{\ell_j}{S},
\qquad
A\le j\le D.
\end{equation}

The generation time is defined by
\begin{equation}
\label{eq_generation}
T
=
\sum_{j=1}^{\omega}jp_j
=
\sum_{j=A}^{D}jp_j
=
\frac1S
\sum_{j=A}^{D}
j\,\ell_j.
\end{equation}

We further define the {\em first temporal moment} by
\begin{equation}
\label{eq_moment}
M
=
\sum_{j=A}^{D}
j\,\ell_j,
\end{equation}
in terms of which \eqref{eq_generation} may be written as
\begin{equation}
\label{eq_generation_time}
T=\frac{M}{S}.
\end{equation}

Define also the {\em logarithmic survivorship dispersion} \(C\) by
\begin{equation}
\label{eq_decline}
C
=
-\sum_{j=A}^{D}
\ell_j\log\ell_j.
\end{equation}
Following Demetrius
\cite{Demetrius78,Demetrius83,Demetrius97},
the evolutionary entropy is defined by
\begin{equation}
\label{eq_eventropy}
H
=
\frac{
-\sum_{j=1}^{\omega}p_j\log p_j
}{
T
}.
\end{equation}

It will be useful to observe that, in this expression for $H$  the numerator is precisely the Shannon entropy of the probability distribution $\{ p_j \}$ 
(\cite{shannon1948,CoverThomas,Gray2011,Jaynes1957,Jaynes1957b}):
\[
\Psi
=
-
\sum_{j=1}^{\omega}
p_j \, 
\log p_j.
\]

From \eqref{eq_p_j} and \eqref{eq_decline}, we have that
$
\Psi = -\sum_{j=A}^{D}
p_j\log p_j
=
\log S
+
\frac1S\,C
$ and thus, in the present context, the entropy assumes the form
\begin{equation}
\label{eq_ENTROPY}
H
=
\frac{
S\log S
+
C
}{
M
}.
\end{equation}

Formula \eqref{eq_ENTROPY} defines the fundamental entropy functional utilized throughout this work, demonstrating that evolutionary entropy is entirely determined by the organization of reproductive distribution across the reproductive interval. The remainder of this paper is devoted to analyzing how different survivorship curves determine the entropy-maximizing reproductive lifespan predicted by this functional.

Biologically, the numerator of \eqref{eq_ENTROPY} measures the dispersion of reproductive contribution across age classes, while the denominator converts this dispersion into a rate by normalizing by generation time. A broad reproductive distribution therefore does not automatically imply high evolutionary entropy: if broadness is achieved only by shifting reproductive mass far into late ages, the increase in Shannon entropy can be offset by the increase in generation time. The optimization problem studied below is precisely this balance between dispersion of reproductive contribution and temporal delay.


\subsection{Open-group Leslie matrices}
\label{sec_open_group}

Many empirical populations do not admit a naturally defined maximal
age class.
In such situations it is convenient to replace the finite Leslie
setting by an open-group formulation in which the final age class
remains reproductively and demographically active.

The simplest open-group Leslie matrix has the form
\[
L=
\begin{bmatrix}
m_1 & m_2 & \cdots & m_{\omega-1} & m_\omega\\
s_1 & 0 & \cdots & 0 & 0\\
0 & s_2 & \cdots & 0 & 0\\
\vdots & \vdots & \ddots & \vdots & \vdots\\
0 & 0 & \cdots & s_{\omega-1} & \rho
\end{bmatrix},
\]
where the coefficient
$
\rho\in[0,1]
$
represents the probability of remaining in the final age class.

Since all survivorship quantities, as in equation \eqref{eq:EL_survivorship}, are rescaled dividing by $\lambda$, the adequate open-group parameter will
\[
q=\frac{\rho}{\lambda},
\]
as we will see below.

This formulation may be interpreted as a finite-dimensional
representation of an infinite Leslie population with a geometric tail.
Indeed, once individuals reach the final class, successive survivorship weights satisfy
\[
\ell_{\omega+k}
=
\ell_\omega q^k,
\qquad k\geq0.
\]
The existence of infinitely many age classes does not imply infinite
survival. For \(q<1\), survivorship decreases geometrically,
$
\ell_{\omega+k}=\ell_\omega q^k,
$
and the resulting series converge. The infinite tail therefore
represents increasingly rare long-lived individuals rather than
immortality.
Alternatively, one may argue that, for any initial population $N_0$, there exists a sufficiently large $k$ such that $N_0 \ell_\omega q^k < 1$, meaning that by cycle $k$ the individuals in class $\omega$ and all subsequent classes have effectively died out.

If the fertility remains constant in the terminal group, that is
\[
m_{\omega+k}=m_\omega,
\qquad k\geq0,
\]
then the reproductive masses beyond age $\omega$ form a
geometric sequence,
\[
\ell_{\omega+k}m_{\omega+k}
=
\ell_\omega m_\omega q^k.
\]
In these conditions, the Euler--Lotka equation becomes
\begin{equation}
\sum_{j=1}^{\omega-1}
\ell_jm_j
+
\ell_\omega m_\omega
\sum_{k=0}^{\infty}
q^k
=
1,
\end{equation}
and therefore, for $q<1$, 
\begin{equation}\label{eq_EL2}
\sum_{j=1}^{\omega-1}
\ell_jm_j
+
\frac{\ell_\omega m_\omega}
{1-q}
=
1.
\end{equation}

The open-group representation is mathematically equivalent to an infinite Leslie matrix with a geometric asymptotic tail. 
Such formulations have been extensively studied in the theory of
infinite Leslie matrices and age-structured populations
\cite{Cus,ElaydiCushing2024,Gosselin,jalves3,jalves2,Demetrius72}.

In the empirical section \ref{sec_Validation} of this paper a substantial fraction (74\%) of the
species analysed belongs naturally to this open-group class.
The geometric-tail representation provides an exact analytical basis for
computing reproductive distributions, generation times and
evolutionary entropy in populations for which reproduction and survival
persist beyond the last explicitly observed age class.

The geometric-tail representation allows the entropy functional to be written, also in this setting, in a completely explicit form.
Under the constant fertility hypothesis, the total survivorship mass becomes, from \eqref{eq_mass},
\[
S_\infty
=
\sum_{j=A}^{\omega-1} \ell_j
+
\ell_\omega
\sum_{k=0}^{\infty}q^k
=
\sum_{j=A}^{\omega-1} \ell_j
+
\frac{\ell_\omega}{1-q}.
\]

\noindent Similarly, the first temporal moment translates from \eqref{eq_moment}  as 
\[
M_\infty
=
\sum_{j=A}^{\omega-1} j\,\ell_j
+
\ell_\omega
\sum_{k=0}^{\infty}
(\omega+k)q^k.
\]
Using the standard identities
$
\sum_{k=0}^{\infty}q^k
=
\frac{1}{1-q},
\, 
\sum_{k=0}^{\infty}kq^k
=
\frac{q}{(1-q)^2},
$
we obtain
\[
M_\infty
=
\sum_{j=A}^{\omega-1} j\,\ell_j
+
\ell_\omega
\left(
\frac{\omega}{1-q}
+
\frac{q}{(1-q)^2}
\right).
\]
Analogously, from \eqref{eq_decline} the logarithmic survivorship dispersion becomes
\[
C_\infty
=
-\sum_{j=A}^{\omega-1}
\ell_j\log\ell_j
-
\ell_\omega
\sum_{k=0}^{\infty}
q^k
\log(\ell_\omega q^k)
\]
which, after elementary series summation, becomes
\[
C_\infty
=
-\sum_{j=A}^{\omega-1}
\ell_j\log\ell_j
-
\frac{\ell_\omega\log\ell_\omega}{1-q}
-
\frac{\ell_\omega q\log q}
{(1-q)^2}.
\]
Thus the corresponding generation time is 
\[
T_\infty
=
\frac{M_\infty}{S_\infty},
\]
whereas the evolutionary entropy assumes the form
\[
H_\infty
=
\frac{
S_\infty\log S_\infty
+
C_\infty
}
{
M_\infty
}.
\]
Comparing with \eqref{eq_generation_time} and \eqref{eq_ENTROPY}, we conclude that 
the entropy functional of an open-group Leslie
population preserves exactly the same structural form as in the
finite-dimensional case. The only difference is that the quantities
\(S\), \(M\) and \(C\) are replaced by their infinite-tail
counterparts \(S_\infty\), \(M_\infty\) and \(C_\infty\), which
incorporate the geometric contribution of the terminal age group.

This observation allows finite and open-group populations to be
treated within a common analytical formalism. In particular, all
entropy calculations performed later in the paper for open-group
species are exact and do not require truncation of the infinite
reproductive tail.


\subsection{Homogeneity and reproductive structure}
In this section we do not assume constant fertilities.

\begin{defn}
Throughout the remainder of the paper we refer to
$
\{\ell_j\}_{1\le j\le \omega}
$
as a EL survivorship sequence, where now
the upper bound \(\omega\) may be finite or infinite, provided the
corresponding Euler--Lotka series converges as in \eqref{eq_EL2}.
\end{defn}

Recall that, as a consequence of the Euler--Lotka  equation \eqref{eq_Euler-Lotka}, the reproductive masses
$
p_j
=
\ell_jm_j
$
define the reproductive distribution whose associated
generation time, by \eqref{eq_moment}, is given by
$
T
=
\sum_{j=1}^{\omega}
j\,p_j,
$
and whose associated evolutionary entropy, by \eqref{eq_eventropy}, is given by
$H
=
\frac{
-\sum_{j=1}^{\omega}
p_j\log p_j
}
{
T
}.$

We now establish a fundamental invariance property of the entropy
functional.

\begin{thm}[Homogeneity of evolutionary entropy]
\label{thm_scale_invariance}

Let
$
\{\ell_j\}_{1\le j\le \omega}
$
be a EL survivorship sequence satisfying
\eqref{eq_Euler-Lotka}.
For any  \(c>0\), define
\[
\widetilde{\ell}_j
=
c\,\ell_j.
\]

Let 
$
\widetilde p_j
=
\frac{
	\widetilde{\ell}_j m_j
}{
	\sum_{k=1}^{\omega}
	\widetilde{\ell}_km_k
}
$
denote the reproductive masses relative to $\widetilde{\ell}_j$, and 
\(\widetilde H\) be the corresponding entropy. 
 
%
Then
$
\widetilde H
=
H.
$

\end{thm}

\begin{proof}

Since
$
\widetilde{\ell}_j
=
c\,\ell_j,
$
we obtain immediately
\[
\widetilde p_j
=
\frac{
c\,\ell_jm_j
}{
\sum_{k=1}^{\omega}
c\,\ell_km_k
}
=
\frac{
\ell_jm_j
}{
\sum_{k=1}^{\omega}
\ell_km_k
}.
\]

From \eqref{eq_Euler-Lotka} it follows that 
$
\sum_{k=1}^{\omega}
\ell_km_k
=
1,
$
and therefore
$
\widetilde p_j
=
\ell_jm_j
=
p_j.
$
Consequently,
\[ 
\widetilde T
=
\sum_{j=1}^{\omega}
j\,\widetilde p_j
=
\sum_{j=1}^{\omega}
j\,p_j
=
T
\quad\]
and
\[
\quad
-\sum_{j=1}^{\omega}
\widetilde p_j
\log\widetilde p_j
=
-\sum_{j=1}^{\omega}
p_j
\log p_j.
\]
Hence
$
\widetilde H
=
H.
$

\end{proof}

Theorem \ref{thm_scale_invariance} shows that evolutionary entropy
depends only on the reproductive distribution and not on
the absolute magnitude of EL survivorship.
Consequently, two EL survivorship sequences that differ only by a positive scaling factor yield identical reproductive distributions, generation times, and evolutionary entropies.
Because the proof depends strictly on the normalization dictated by the Euler–Lotka equation, this result remains valid for finite, open-group, and infinite Leslie matrices alike, independent of the finiteness of \(\omega\).



\begin{rem}
Theorem~\ref{thm_scale_invariance} shows that both evolutionary
entropy and generation time are homogeneous functions of degree zero with
respect to uniform rescaling of survivorship or fertility. More
precisely, denoting by $\mathcal{M} = \{m_j\}_{j=1}^\omega$, 
$\mathcal{L} = \{\ell_j\}_{j=1}^\omega$ respectively the fertilities and the EL survivorships, we have that, 
for any constants \(a,c>0\),
\[
H(a\,\mathcal{M},c\,\mathcal{L})
=
H(\mathcal{M},\mathcal{L}),
\qquad
T(a\,\mathcal{M},c\,\mathcal{L})
=
T(\mathcal{M},\mathcal{L}).
\]

	\noindent Consequently, both quantities depend only on the normalized
reproductive distribution and not on the absolute scale of fertility
or survivorship.
\end{rem}

\emph{This observation indicates that evolutionary entropy is fundamentally
a measure of reproductive organization rather than reproductive
quantity.}


\FloatBarrier

\section{Generalized exponential survivorship model}\label{sec_exp_model}

Families of survivorship functions have been extensively employed in
demography, ecology and reliability theory, including Gompertz-type
and Weibull-type survival models, matrix population approaches, and
recent studies of Demetrius-Keyfitz' entropy and its discrete-time
formulations
\cite{Bernard2024,Gompertz1825,Weibull1951,Caswell,VAU,Baudisch2008,
	Jones2014,Fernandez2015,Colchero2016,deVries2023,Giaimo2024}.

As a general model encompassing several classical survivorship
patterns, we consider the two-parameter family
\begin{equation}
\label{eq_exponential}
l_j=
\exp\left(
-\alpha
\left(
\frac{j-1}{\omega-1}
\right)^p
\right),
\qquad
j=1,\dots,\omega,
\end{equation}
where $\alpha>0$ and $p>0$. The normalization ensures $l_1=1$ and $l_\omega=e^{-\alpha}$, with $\alpha$ dictating the global survivorship scale and $p$ governing the profile of mortality decay across the life cycle.

While traditional demography uses survivorship functions to describe complete life histories, our focus shifts from these curves to the reproductive distributions they induce. Distinct survivorship functions can generate identical reproductive distributions and entropies; conversely, demographically similar curves can yield vastly different reproductive structures once population growth or decline is incorporated. 

Consequently, our analysis separates the physical shape of the survivorship function from the reproductive distribution it generates after demographic rescaling via the Perron eigenvalue $\lambda$. Because evolutionary entropy depends strictly on the latter, standard survivorship classifications based on juvenile or pre-reproductive mortality are only indirectly relevant to classifying reproductive strategies.

The following section formalizes this correspondence, deriving the entropy formulas for the principal reproductive regimes generated by family \eqref{eq_exponential}.

\subsection{Entropy structure of the generalized exponential model}

We now show how the general exponential model 
\eqref{eq_exponential} relates to the standard survivorship functions in the literature. 
Throughout this section, fertility
is assumed to be constant throughout the reproductive interval.
\subsubsection{Geometric survivorship}
\label{subsub_geometric}

For
$
p=1,
$
we obtain
\begin{equation}
\label{eq_l_j}
l_j
=
\exp\left(
-\alpha
\frac{j-1}{\omega-1}
\right).
\end{equation}
Thus the transition survivorships are constant, defining a transition constant $\rho$:
\begin{equation}
\label{eq_geometric_TypeII}
s_j
=
\frac{l_{j+1}}{l_j}
=
\exp\left(
-\frac{\alpha}{\omega-1}
\right)
\stackrel{\text{def}}{\equiv} \rho.
\end{equation}
This is precisely the classical Type II survivorship regime.

\subsubsection{Gaussian survivorship}

For
$
p=2
$
the generalized exponential model reduces to
\[
l_j
=
\exp\left(
-\alpha
\left(
\frac{j-1}{\omega-1}
\right)^2
\right).
\]

This is precisely the discrete analogue of the Gaussian Type III
survivorship curve introduced by Demetrius \cite{Demetrius78}, $
l(x)
=
\exp\left(
-\frac{\pi}{4e_0^2}x^2
\right),$ where $e_0$ is the life expectancy.
Identifying coefficients yields
$
\alpha
=
\frac{\pi(\omega-1)^2}{4e_0^2}$.
The corresponding transition survivorships are
\[
s_j
=
\frac{l_{j+1}}{l_j}
=
\exp\left(
-\frac{\pi}{4e_0^2}
(2j-1)
\right).
\]

Thus transition survivorship decreases exponentially with age.
Equivalently, the discrete hazard
$
h_j=-\log s_j
$
satisfies
$
h_j
=
\frac{\pi}{4e_0^2}
(2j-1)$, 
showing that the hazard increases linearly with age.

\subsubsection{Rectangular survivorship}
\label{subsub_rectangular}
In the limit
$
p\to\infty,
$
we have, for every \(1\le j<\omega\),
\[
\left(
\frac{j-1}{\omega-1}
\right)^p
\longrightarrow 0 
\]
uniformly in $j$ and hence
\[
l_j\longrightarrow 1 \quad {\rm uniformly}.
\]
This corresponds to the classical  Type I rectangular survivorship limit.


\subsection{Canonical reproductive distributions}
\label{sub_canonical}
The preceding survivorship functions induce reproductive
distributions through the EL survivorship weights
\(\ell_j\). .

%

Recall that the reproductive interval is $A,\ldots, D$. 
For ease of notation, in the remainder of this section we denote by $n=D-A+1$
the length of the reproductive interval and by $\mu = \dfrac{A+D}{2}$ its midpoint.

\begin{prop}[Uniform reproductive distribution]
	\label{prop_uniform}
Assume that the EL survivorship weights are constant
throughout the reproductive interval. Then the
reproductive distribution is uniform and the evolutionary entropy is
\begin{equation}
\label{eq_Type_I_entropy}
H
=\frac{
\log n
}{ \mu}.
\end{equation}
\end{prop}

\begin{proof}
Since the EL survivorship weights are constant throughout
\(A,\dots,D\), normalization gives
\[
p_j=\frac1n,
\qquad j=A,\dots,D.
\]
Moreover, the relevant parameters $S, \, M, \,  C$ take the values
\[
S=n,
\qquad
M=\sum_{j=A}^{D}j
=
\frac{(A+D)n}{2} = \mu \, n,
\qquad
C=0.
\]
Substitution into \eqref{eq_ENTROPY} gives
\[
H
=
\frac{\log n}{M/S}
=
\frac{
\log n
}{
\dfrac{A+D}{2} 
} = \frac{\log n}{\mu},
\]
as claimed.
\end{proof}

This distribution corresponds to the maximal spread of
reproductive contributions across the reproductive interval.
It is obtained, for example, from the rectangular limit
\(p\to\infty\) in the stationary case \(\lambda=1\). It is
also obtained from the small-exponent limit \(p\to0^+\)
when \(\lambda=1\).

The next result deals with the case where the
EL survivorship weights form a finite geometric progression with ratio \(q <1\), along the reproductive interval \(A,\dots,D\).
For ease of notation we relabel the indices $j$ of the age classes to $k=j-A$.
This relabeling leads to the expressions
\begin{equation}
	\ell_{A+k}
	=
	\left(
	\frac{\rho}{\lambda}
	\right)^k=q^k,
	\qquad
	0<q<1,
	\label{eq:typeII_survivorship}
\end{equation}
where $k=0,\ldots,n-1$ and $n=D-A+1.$

\begin{prop}[Finite geometric reproductive distribution]
	\label{prop_geometric}
Suppose the
EL survivorship weights form a finite geometric progression
with ratio \(q <1\) along the reproductive interval \(A,\dots,D\). Then the evolutionary entropy is
\begin{equation}
\label{eq_typeII_entropy_q}
H_n
=
\frac{
-\log q\,
\left[
\displaystyle
\frac{q}{1-q}
-
\frac{nq^n}{1-q^n}
\right]
+
\log\left(
\dfrac{1-q^n}{1-q}
\right)
}{
A
+
\displaystyle
\frac{q}{1-q}
-
\displaystyle
\frac{nq^n}{1-q^n}
}.
\end{equation}
\end{prop}

\begin{proof}
	We first rescale the EL survivorships $\ell_j$ so that 
	
	\begin{equation}
		\label{eq_typeII_discounted_profile}
		\ell_{A+k}
		=
		q^k,
		\qquad
		k=0,\ldots,n-1;
	\end{equation}
by the homogeneity property (Theorem \ref{thm_scale_invariance}) this normalization leaves the entropy unchanged.
	Routine computations then show that 
\[
S_n
=
\sum_{k=0}^{n-1}q^k
=
\frac{1-q^n}{1-q},
\]
and
\[
M_n
=
\sum_{k=0}^{n-1}(A+k)q^k
=
A\frac{1-q^n}{1-q}
+
\frac{
q-nq^n+(n-1)q^{n+1}
}{
(1-q)^2
}.
\]
Therefore the generation time is given by
\[
T_n
=
\frac{M_n}{S_n}
=
A
+
\frac{q}{1-q}
-
\frac{nq^n}{1-q^n}.
\]
The reproductive distribution is
\[
p_{A+k}
=
\frac{q^k}{S_n},
\qquad
k=0,\ldots,n-1.
\]
As noted, this is a probability distribution. Computing its Shannon entropy we obtain
\[
\Psi_n
=
-\sum_{k=0}^{n-1}p_{A+k}\log p_{A+k}
=
-\log q\,(T_n-A)+\log S_n.
\]
Recalling from \eqref{eq_eventropy} that
\(H_n=\Psi_n/T_n\) proves \eqref{eq_typeII_entropy_q}.
\end{proof}

For the geometric survivorship model \(p=1\), the transition
survivorships are constant and equal to \(\rho\). Hence, the induced 
geometric reproductive ratio is
\begin{equation}
\label{eq_q_rho_lambda}
q
=
\frac{\rho}{\lambda}.
\end{equation}

\begin{prop}[Infinite geometric reproductive distribution]
Let $D=\infty$. Suppose that, from age class \(A\) onward, the EL survivorship
weights form an infinite geometric progression with ratio
\(q<1\). Then 
the evolutionary entropy is given by
\begin{equation}
\label{eq_typeII_open_entropy_q_compact}
H_\infty
=
-\frac{
(1-q)\log(1-q)
+
q\log q
}{
A(1-q)+q
}.
\end{equation}
\end{prop}

\begin{proof}
	As in Proposition \ref{prop_geometric},  the homogeneity property of the entropy implies we may assume without loss of generality that 
\[
\ell_{A+k}
=
q^k,
\qquad
k\ge0.
\]
Therefore
\[
S_\infty
=
\sum_{k=0}^{\infty}q^k
=
\frac1{1-q},
\]
and
\[
M_\infty
=
\sum_{k=0}^{\infty}(A+k)q^k
=
\frac{A}{1-q}
+
\frac{q}{(1-q)^2}.
\]
Thus
\[
T_\infty
=
\frac{M_\infty}{S_\infty}
=
A+\frac{q}{1-q}.
\]

\noindent The reproductive distribution is the geometric
distribution
\[
p_k^{(\infty)}
=
(1-q)q^k,
\qquad
k\ge0, 
\]
whose Shannon entropy $\Psi_\infty$ is given by
\[
\Psi_\infty
=
-\log(1-q)
-
\frac{q}{1-q}\log q.
\]
Since from \eqref{eq_eventropy} \(H_\infty=\Psi_\infty/T_\infty\), we finally obtain
\[
H_\infty
=
-\frac{
(1-q)\log(1-q)
+
q\log q
}{
A(1-q)+q
},
\]
as claimed.
\end{proof}


\subsection{From survivorship functions to reproductive distributions}

The same reproductive distribution may arise from distinct survivorship functions.

For the geometric survivorship model, we have
\[
l_{A+k}\propto \rho^k,
\]
and therefore
\[
\ell_{A+k}
=
\frac{l_{A+k}}{\lambda^{A+k}}
\propto
\left(
\frac{\rho}{\lambda}
\right)^k.
\]
Thus the induced reproductive distribution is geometric with effective ratio
${\displaystyle q=\frac{\rho}{\lambda}.} $
On the other hand, in the singular limit
$
p\to0^+,
$
we have
\[
l_j\longrightarrow e^{-\alpha} \quad {\rm uniformly}.
\]
Consequently,
\[
\ell_j
=
\frac{l_j}{\lambda^j}
\propto
\lambda^{-j}.
\]
Thus the induced reproductive distribution is again geometric, with effective ratio $
q=\frac1\lambda$.

We assume $
0<\rho/\lambda<1$.
Indeed, in the presence of an infinite tail, if $
\frac{\rho}{\lambda}\ge 1$,
the geometric series diverges, yielding an infinite total offspring contribution. Consequently, in infinite-age models, including open-group Leslie models, this regime cannot correspond to a biologically realizable population.

The entropy-maximization results developed in the following sections
will therefore be formulated in terms of the reproductive parameter
\(q\), rather than in terms of the particular survivorship function
that generates it.

\FloatBarrier

\section{Entropy maximization and reproductive lifespan}
\label{sec_entropy_maximization}

The previous sections show that evolutionary entropy depends
only on the reproductive masses. The central optimization problem therefore becomes the
following:

\begin{quote}
For fixed reproductive maturity age \(A\), what is the reproductive
endpoint \(D\) for which evolutionary entropy is maximized?
\end{quote}

Throughout this section the EL survivorship sequence $
\ell_j
=
{\displaystyle \frac{l_j}{\lambda^j},}
\ j\ge A,
$
is regarded as fixed, while the reproductive endpoint
$
D\ge A
$
is allowed to vary. Thus all the relevant variables will become functions of $D$, namely:

\begin{enumerate}
	\item the total survivorship mass $S_D
	=	\sum_{j=A}^{D}\ell_j,$
	\item  the first temporal moment $M_D
	=
	\sum_{j=A}^{D} j\,\ell_j,$
	\item the logarithmic reproductive dispersion $ C_D
	=
	-\sum_{j=A}^{D}
	\ell_j\log \ell_j$,
	\item the reproductive distribution $
		p_j^{(D)}
		=
		\frac{\ell_j}{S_D}
		=
		\frac{
			l_j\lambda^{-j}
		}{
			\sum_{k=A}^{D}
			l_k\lambda^{-k}
		},$
		\item the Shannon entropy $\Psi_D
		=
		-
		\sum_{j=A}^{D}
		p_j^{(D)}
		\log p_j^{(D)}$,
		\item the generation time 
$T_D
=
\sum_{j=A}^{D}
j\,p_j^{(D)}
=
\frac{M_D}{S_D}$.
\end{enumerate}

With these notations,  evolutionary entropy may be written in terms
of Shannon entropy as
\begin{equation}
	\label{eq_EE_Shannon}
	H(D)
	=
	\frac{\Psi_D}{T_D}
\end{equation}
or equivalently, using equation \eqref{eq_ENTROPY}, as
\begin{equation}
	\label{eq:HD-basic}
	H(D)
	=
	\frac{
		S_D\log S_D
		+
		C_D
	}{
		M_D
	}.
\end{equation}
Both forms will be helpful in what follows.

%
%
%
%


The corresponding Shannon entropy
\cite{shannon1948,CoverThomas,Gray2011,Jaynes1957,Jaynes1957b}
is
\[
\Psi_D
=
-
\sum_{j=A}^{D}
p_j^{(D)}
\log p_j^{(D)}.
\]

%
%


\subsection{Asymptotic structure of evolutionary entropy}

We recall the classical upper bound for discrete
Shannon entropy with fixed mean (see, e.g.,
\cite{CoverThomas}). Given a probability distribution
$
q=\{q_j\}_{j\geq0}
$
on \(\mathbb N_0\) with mean
$
\mu
=
\sum_{j=0}^{\infty} j q_j,
$
then
\begin{equation}
\label{eq_entropybound}
-
\sum_{j=0}^{\infty}
q_j\log q_j
\leq
(\mu+1)\log(\mu+1)-\mu\log\mu,
\end{equation}
with equality holding if and only if \(q\) is a geometric distribution.

\begin{thm}[Asymptotic entropy collapse]
\label{thm:entropy-collapse}
Let \(D\to\infty\). Then
\[
T_D\to\infty
\qquad\mbox{\rm if and only if}\qquad
M_D\to\infty.
\]
Moreover, each of these conditions implies
$
H(D)\to0.
$
\end{thm}

\begin{rem}
We refer to this phenomenon as {\em asymptotic collapse of evolutionary entropy}.
\end{rem}

\begin{proof}

We first show that
$
M_D\to\infty
\Longrightarrow
T_D\to\infty.
$
Fix \(N\geq A\). For every \(D>N\),
\[
M_D
=
\sum_{j=A}^{N} j\ell_j
+
\sum_{j=N+1}^{D} j\ell_j.
\]

Since \(j\geq N+1\) in the second sum, it follows that
\[ 
M_D
\geq
\sum_{j=A}^{N} j\ell_j
+
(N+1)
\sum_{j=N+1}^{D}\ell_j.
\]

Dividing through by \(S_D\),
\[
T_D
=
\frac{M_D}{S_D}
\geq
\frac1{S_D}
\sum_{j=A}^{N} j\ell_j
+
(N+1)
\left(
1-
\frac1{S_D}
\sum_{j=A}^{N}\ell_j
\right).
\]

\noindent Suppose first that \(S_D\to+\infty\). Then
$
\frac1{S_D}
\sum_{j=A}^{N} j\ell_j
\to0 $, implying 
$
\frac1{S_D}
\sum_{j=A}^{N}\ell_j
\to0, 
$ and therefore
\[
\liminf_{D\to+\infty}T_D
\geq
N+1.
\]
Since \(N\) is arbitrary, this condition  implies
$
T_D\to+\infty.
$

To prove the converse implication, suppose that \(S_D\) is bounded.
Since \(S_D\) is increasing, it converges to some limit $
S_\infty>0.$ 
Since 
$
T_D
=
\frac{M_D}{S_D},
$
we conclude that \(T_D\to\infty\) only if \(M_D\to\infty\). This concludes the proof of the equivalence
\[
T_D\to+\infty
\quad\Longleftrightarrow\quad
M_D\to+\infty.
\]

We now show that either of these conditions implies the asymptotic vanishing of entropy.

First note that, since \(p^{(D)}\) is a probability distribution on
\(\mathbb N_0\) with mean \(T_D\), the entropy bound
\eqref{eq_entropybound} yields
\[
\Psi_D
\leq
(T_D+1)\log(T_D+1)
-
T_D\log T_D.
\]

Suppose now that $T_D \to \infty$ as $D \to \infty$. Then 
\[
(T_D+1)\log(T_D+1)
-
T_D\log T_D
=
\log T_D+1+o(1).
\]

Therefore
$
\Psi_D
=
O(\log T_D),
$
and consequently
\[
H(D)
=
\frac{\Psi_D}{T_D}
=
O\!\left(
\frac{\log T_D}{T_D}
\right),
\]

showing that
$
H(D)\to0
$. This concludes the proof.

\end{proof}

We observe that if only one effective reproductive class exists,
say \(m_k\neq0\) for a single \(k\), then the Shannon and evolutionary entropies are
zero. This singular case must be trivially excluded from the hypotheses of our next result.

\begin{cor}[Existence of a finite entropy-maximizing endpoint]
\label{cor:finite-maximizer}

Suppose that
$
M_D\to+\infty$
as $
(D\to+\infty)
$
and assume that at least two reproductive age classes have positive reproductive masses.
Then evolutionary entropy attains a finite global maximum. 
More precisely, there exists
$
D^\ast>A
$
such that
\[
H(D^\ast)
=
\max_{D\geq A}H(D).
\]

\end{cor}

\begin{proof}

Since at least two reproductive age classes have positive
reproductive masses, there exists \(D_0>A\)
such that \(p^{(D_0)}\) is not concentrated at a single age class.
Hence $
\Psi_{D_0}>0$ and $
H(D_0)>0.$
On the other hand
$
H(A)=0,
$ which implies
\[
\sup_{D\geq A}H(D) \geq H(D_0) >0.
\]
Since $M_D\to+\infty$, it follows from 
Theorem~\ref{thm:entropy-collapse} that 
$
H(D)\to0$ as 
$D\to+\infty.$
Hence there exists \(N\geq A\) such that
\[
H(D)
<
\frac12
\sup_{E\geq A}H(E)
\qquad
\text{for all }
D>N.
\]
Consequently,
\[
\sup_{D\geq A}H(D)
=
\max_{A\leq D\leq N}H(D).
\]
Since
$
J = \{A,A+1,\ldots,N\} $
is finite, the maximum is attained at some  \(D^\ast \in J \).
Because \(H(A)=0\) and the global maximum is strictly positive,
it follows necessarily that $
D^\ast>A.
$

\end{proof}

\begin{rem}[Divergent and finite first-moment regimes]
\label{thm:first-moment-dichotomy}

We now summarize the distinct entropy regimes determined by the convergence behaviour of the first moment $M_D$.

Consider the complete EL survivorship sequence
$
\{\ell_j\}_{j\geq A}.
$

\begin{enumerate}

\item
If
$
\sum_{j=A}^{\infty}
j\ell_j
=
+\infty,
$
then
$
M_D\to+\infty,$
\, 
$T_D\to+\infty,$
\, 
$H(D)\to0.$
If, in addition, there is more than one reproductive age class
with positive reproductive masses, then \(H(D)\)
possesses a finite global maximizer $D^*$.

\item
If 
$
\sum_{j=A}^{\infty}
j\ell_j
<
+\infty,
$
then
$
\sum_{j=A}^{\infty}
\ell_j
<
+\infty
$
and the series defining $S_D$ are convergent, with  
$
S_D\to S_\infty<+\infty,$
\, 
$M_D\to M_\infty<+\infty.$
Consequently,
\[
T_D
=
\frac{M_D}{S_D}
\longrightarrow
T_\infty
=
\frac{M_\infty}{S_\infty}
<
+\infty.
\]
Defining the probability distribution
\[
p_j^{(\infty)}
=
\frac{\ell_j}{S_\infty},
\qquad
j\geq A,
\]
the associated Shannon entropy is
\[
\Psi_\infty
=
-
\sum_{j=A}^{\infty}
p_j^{(\infty)}
\log p_j^{(\infty)}.
\]
If
$
\Psi_\infty<+\infty,
$
then
\[
H(D)
\longrightarrow
H_\infty
:=
\frac{
S_\infty\log S_\infty
-
\sum_{j=A}^{\infty}
\ell_j\log\ell_j
}{
M_\infty
}.
\]

Within this regime, the existence of a finite interior maximizer is not guaranteed. 
Depending on the asymptotic structure of the tail of the distribution, $H(D)$ may either exceed its 
limiting value $H_\infty$ at a finite point $D^*$  or converge monotonically toward $H_\infty$.


\end{enumerate}

\end{rem}

%
%
%
%

\begin{rem}[Natural extension to a continuous variable]
\label{rem_extensions}
The quantities $S_D$, $M_D$, $C_D$, $T_D$ and $H(D)$ are first defined for integer endpoints $D$. In the two optimization regimes treated explicitly below, however, these quantities are given by closed-form expressions depending on $D$ through elementary functions such as $D-A+1$, $q^{D-A+1}$, and logarithms. We use the phrase \emph{natural extension} for the real-variable functions obtained from these closed forms by replacing the integer endpoint by a real variable.

This convention is not an arbitrary interpolation of a discrete sequence. It is the continuous version of the same formula used to compute the integer values. Nor do we require, or claim, uniqueness of analytic interpolation from values prescribed only at the integers. The continuous functions are used only to locate and classify candidate extrema; the admissible demographic optimizer is always obtained by evaluating $H(D)$ on integer reproductive endpoints. Thus the continuous calculation provides a rigorous and transparent way to identify the finite set of integer endpoints that must be checked in the original Leslie problem.
\end{rem}

%
%
%
%
%
%
%
%
%
%


\subsection{Entropy maximization in the rectangular limit}

In the rectangular limit, the reproductive
distribution is uniform on the reproductive interval.
Recall that from Proposition \ref{prop_uniform}
we have
\[
H(D)
= \frac{\log n}{\mu} =
\frac{\log(D-A+1)}{\dfrac{A+D}{2}} .
\]

For fixed maturity age \(A\), according to the principle of natural
extensions described in Remark~\ref{rem_extensions}, we substitute the integer parameter $D$ by the real variable \(x\geq A\), and define the corresponding
real-variable function
\begin{equation}
\label{eq_F_function}
F(x)
=
\frac{\log(x-A+1)}{\dfrac{A+x}{2}} .
\end{equation}

Upon differentiation we obtain
\[
F'(x)
=
\frac{
2\left( 
\dfrac{A+x}{x-A+1}
-
\log(x-A+1)
\right)
}{
(A+x)^2
},
\]
and therefore critical points of \(F\) satisfy
\begin{equation}
\label{eq_critical}
\frac{A+x}{x-A+1}
=
\log(x-A+1).
\end{equation}

\noindent Introducing the variable
$
u=x-A+1,
$ this 
equation becomes
\[
1+\frac{2A-1}{u}
=
\log u,
\]
or equivalently
\[
u\log u-u
=
2A-1.
\]
With the further change of variable \(u=e^y\), we obtain
$
(y-1)e^{y-1}
=
\frac{2A-1}{e},
$
which, in terms of the Lambert \(W\) function
\cite{Corless1996}, is written 
\[
y-1
=
W\!\left(
\frac{2A-1}{e}
\right).
\]
Exponentiating and returning to the original variables, we obtain
the continuous rectangular-limit entropy maximizer
\begin{equation}
\label{eq_Type_I_maximizer}
D^\ast
=
A-1+
e^{1+W\left(\frac{2A-1}{e}\right)}.
\end{equation}

In terms of the original discrete demographic problem, the admissible
optimizer is the integer nearest to $D^*$   in
$
\{A,\ldots,\omega\}.
$

\subsection{The geometric reproductive regime}

Recall from \ref{subsub_geometric} that the geometric regime is characterized by geometric decay by a constant ratio
$
q=\frac{\rho}{\lambda}
$
arising from geometric post-maturity survivorship
\eqref{eq_geometric_TypeII}.

As in Section \ref{sub_canonical}, for ease of notation we relabel the indices $j$ of the age classes to $k=j-A$, writing

\begin{equation}
\ell_{A+k}
=
\left(
\frac{\rho}{\lambda}
\right)^k=q^k,
\label{eq:typeII_survivorshipA}
\end{equation}
where $k=0,\ldots,n-1$ and $n=D-A+1$ is the length of the reproductive window.

Although this behaviour coincides with the classical Type II
survivorship regime when $\lambda=1$, it is by no means restricted
to that setting. During the reproductive period, populations
exhibiting Type I, Type II or even Type III survivorship patterns
are often well approximated by either approximately constant adult
hazard rates or very low mortality rates. Consequently, geometric
or near-geometric reproductive masses arise naturally across
a broad range of demographic regimes.

Define the normalization factor
\begin{equation}
Z(n)
=
\sum_{k=0}^{n-1}
q^k
=
\frac{
1-q^n
}{
1-q
},
\label{eq:typeII_partition}
\end{equation}
and the reproductive distribution
\begin{equation}
q_k^{(n)}
=
\frac{
q^k
}{
Z(n)
},
\qquad
k=0,\ldots,n-1,
\label{eq:typeII_distribution}
\end{equation}
which is of course a probability distribution.
The mean reproductive displacement relative to maturity is 
\begin{equation}
m(n)
=
\sum_{k=0}^{n-1}
k\,q_k^{(n)}
\label{eq:typeII_mean_definition}
\end{equation}
and the generation time is therefore given by
\begin{equation}
T(n)
=
A+m(n).
\label{eq:typeII_generation_time}
\end{equation}

\begin{lem}[Mean of the truncated geometric distribution]
For the geometric reproductive distribution
\eqref{eq:typeII_distribution}, we have
\begin{equation}
m(n)
=
\frac{q}{1-q}
-
\frac{nq^n}{
1-q^n}.
\label{eq:typeII_mean}
\end{equation}
\end{lem}

\begin{proof}
Using the classical finite geometric identity
\begin{equation}
\sum_{k=0}^{n-1}kq^k
=
\frac{q-nq^n+(n-1)q^{n+1}}
{(1-q)^2}
\label{eq:typeII_geometric_sum}
\end{equation}

\noindent and combining
\eqref{eq:typeII_partition},
\eqref{eq:typeII_distribution}
and
\eqref{eq:typeII_mean_definition},
we obtain
\begin{equation}
m(n)
=
\frac{
q-nq^n+(n-1)q^{n+1}
}{
(1-q)(1-q^n)
}.
\label{eq:typeII_mean_intermediate}
\end{equation}

\noindent The numerator on the right hand-side 
may be written as
$
q(1-q^n)-nq^n(1-q),
$
whereupon substitution 
yields
\eqref{eq:typeII_mean}, proving the lemma.
\end{proof}

In the sequel it will be sometimes convenient to use the auxiliary parameter
\begin{equation}
\beta=-\log q>0.
\label{eq:typeII_beta}
\end{equation}

\begin{lem}[Entropy of the truncated geometric distribution]
	\label{lemma_truncated}
The Shannon entropy of the reproductive distribution
\eqref{eq:typeII_distribution}
is
\begin{equation}
\Psi(n)=\beta m(n)+\log Z(n).
\label{eq:typeII_shannon}
\end{equation}
\end{lem}

\begin{proof}
From
\eqref{eq:typeII_distribution},
we have
$
\log q_k^{(n)}
=
k\log q-\log Z(n).
$
Using the definition of Shannon entropy, we obtain
\begin{align}
\Psi(n)
&=
-\sum_{k=0}^{n-1}
q_k^{(n)}\log q_k^{(n)}\\
&=
-\log q
\sum_{k=0}^{n-1}kq_k^{(n)}
+
\log Z(n)
\sum_{k=0}^{n-1}q_k^{(n)}.
\label{eq:typeII_entropy_expansion}
\end{align}

Since
\(
\{q_k^{(n)}\}
\)
is a probability distribution,
$
\sum_{k=0}^{n-1}q_k^{(n)}=1,
$
whereby the last term in
\eqref{eq:typeII_entropy_expansion}
reduces to \(\log Z(n)\). On the other hand, using
\eqref{eq:typeII_mean_definition},
we obtain
\[
\Psi(n)
=
-\log q\,m(n)+\log Z(n).
\]
Finally, using the definition of $\beta$ in
\eqref{eq:typeII_beta},
we obtain
\[
\Psi(n)
=
\beta m(n)+\log Z(n).
\]
as claimed. 
\end{proof}
In view of Lemma \ref{lemma_truncated} and the expression \eqref{eq_EE_Shannon} of evolutionary entropy 
in terms of the Shannon entropy,
we obtain
\begin{equation}
H(n)
=
\frac{
\beta m(n)+\log Z(n)
}{
A+m(n)
}.
\label{eq:typeII_entropy}
\end{equation}

To further analyse the structure of the entropy function, we use the closed-form natural extension described in Remark~\ref{rem_extensions}. Thus the integer variable \(n\ge1\) is replaced by a real variable \(x\ge1\), and we define
the normalization
\begin{equation}
Z(x)
=
\frac{1-e^{-\beta x}}{1-q},
\label{eq:typeII_Zxi}
\end{equation}
the mean reproductive displacement
\begin{equation}
m(x)
=
\frac{q}{1-q}
-
\frac{x e^{-\beta x}}
{1-e^{-\beta x}},
\label{eq:typeII_mxi}
\end{equation}
the generation time
\begin{equation}
T(x)=A+m(x),
\label{eq:typeII_Txi}
\end{equation}
and the evolutionary entropy
\begin{equation}
H(x)
=
\frac{
\beta m(x)+\log Z(x)
}{
T(x)
}.
\label{eq:typeII_Hxi}
\end{equation}

We also define two auxiliary functions which will play an important role below. The first is the function \(F\) defined by
\begin{equation}
F(x)
=
H(x)-\beta
=
\frac{
\log Z(x)-\beta A
}{
T(x)
},
\label{eq:typeII_F}
\end{equation}
and the second is the function \(\Phi(x)\) defined by
\begin{equation}
\Phi(x)
=
\frac{
\beta(1-e^{-\beta x})
}{
\beta x-1+e^{-\beta x}
}.
\label{eq:typeII_Phi}
\end{equation}

\begin{lem}
\label{lem_derivative_F}
The function 
\(
F
\)
satisfies the differential identity
\begin{equation}
F'(x)
=
\frac{T'(x)}{T(x)}
\left[
\Phi(x)-F(x)
\right].
\label{eq:typeII_Fprime}
\end{equation}

\end{lem}

\begin{proof}
Differentiating $F$ from 
\eqref{eq:typeII_F}, we obtain
\begin{align}
F'(x)
&=
\frac{
(\log Z(x))' \, T(x)
-
(\log Z(x)-\beta A) \, T'(x)
}{
T(x)^2
}
\notag\\
&=
\frac{T'(x)}{T(x)}
\left[
\frac{(\log Z(x))'}{T'(x)}
-
F(x)
\right].
\label{eq:typeII_Fprime_intermediate}
\end{align}
Taking logs in \eqref{eq:typeII_Zxi} and differentiating, we obtain 
\begin{equation}
(\log Z(x))'
=
\frac{
\beta e^{-\beta x}
}{
1-e^{-\beta x}
},
\label{eq:typeII_logZprime}
\end{equation}
while differentiating $m(x)$ from 
\eqref{eq:typeII_mxi}
and using
\eqref{eq:typeII_Txi}
yields
\begin{equation}
T'(x)
=
\frac{
e^{-\beta x}
(
\beta x-1+e^{-\beta x}
)
}{
(1-e^{-\beta x})^2}.
\label{eq:typeII_Tprime}
\end{equation}
Substituting these two expressions into 
\eqref{eq:typeII_Fprime_intermediate}
yields
\eqref{eq:typeII_Fprime}, proving the Lemma.
\end{proof}

\begin{lem}
\label{lem_positivity}
 For all $x > 0$, we have 
$ T'(x)>0  \quad {\rm and }  \  
\Phi'(x)<0. $
\end{lem}

\begin{proof}
Consider the elementary inequality
$u-1+e^{-u}>0$, valid for $u > 0$.
Taking 
\(
u=\beta x
\)
and substituting in \eqref{eq:typeII_Tprime}  shows that $T'(x)>0$.

We next show that $\Phi'(x)<0$. Upon differentiating $\Phi$ from 
\eqref{eq:typeII_Phi}, we obtain
\begin{equation}
\Phi'(x)
=
\frac{
\beta^2
\left[
e^{-\beta x}(\beta x+1)-1
\right]
}{
(
\beta x-1+e^{-\beta x}
)^2}.
\label{eq:typeII_Phiprime}
\end{equation}

\noindent Consider the elementary inequality $e^u>u+1$, valid for $u>0$, in the equivalent form $ e^{-u}(u+1)<1.$ Taking 
$
u =\beta x
$
and applying it to the expression above
shows that $\Phi'(x)<0$  for all $x>0$.
\end{proof}

\begin{thm}[Uniqueness of the Type II critical point]\label{Unimodal}
The evolutionary entropy $H(x)$ 
associated with Type II  survivorship possesses at
most one critical point on
\(
[1,\infty).
\)
Moreover, if a critical point of $H$ exists, it  is a strict local maximum.
\end{thm}

\begin{proof}
It will be useful to recall, from \eqref{eq:typeII_F}, the expression for the entropy $H$
\[ H(x) = \beta+F(x) \]
and, from  Lemma \ref{lem_derivative_F}, the differential identity
\[
F'(x)
=
\frac{T'(x)}{T(x)}
[\Phi(x)-F(x)]. \]

\noindent Let the term in square brackets  define the function  $G(x)=\Phi(x) - F(x).$
\noindent The differential identity for \(F\) is then written in terms of $G$ as
\[
F'(x)
=
\frac{T'(x)}{T(x)}G(x).
\]
The critical points of \(H\) satisfy
$
H'(x)=F'(x)=0.
$
Note that, since $T(x) > 0$ by definition and  $T'(x) > 0$ by Lemma \ref{lem_positivity}, the critical points of $H$ are precisely the zeros of \(G\). 

Differentiating $G$, we obtain
\begin{equation}
G'(x)
=
\Phi'(x)-\frac{T'(x)}{T(x)}G(x).
\label{eq:typeII_Gprime}
\end{equation}
By Lemma~\ref{lem_positivity}, $T'(x)>0$ and $\Phi'(x)<0$. Hence, at any zero $x^\ast$ of $G$,
\[
G'(x^\ast)=\Phi'(x^\ast)<0.
\]
Thus every zero of $G$ is simple and is crossed from positive values to negative values. Since
\[
H'(x)=F'(x)=\frac{T'(x)}{T(x)}G(x)
\]
and $T'(x)/T(x)>0$, the sign of $H'$ is the sign of $G$. Consequently every critical point of $H$ is a strict local maximum.

It remains to exclude two such critical points. Suppose that $x_1<x_2$ were two zeros of $G$. Since $G$ crosses from positive to negative at $x_1$, we have $G<0$ immediately to the right of $x_1$. In order to cross again from positive to negative at $x_2$, continuity would force $G$ to have changed from negative to positive at some intermediate zero, contradicting the fact just proved that every zero is crossed from positive to negative. Hence $G$ has at most one zero, and $H$ has at most one critical point. If that point exists, it is a strict local maximum.

%
%
%
\end{proof}

The previous theorem establishes uniqueness of the entropy maximum whenever such a maximum exists. It does not,  however, ensure the existence of a maximum, which in general may not exist. Our next results determine the boundary separating the regimes of existence of a maximum and 
monotone growth of the entropy.

\begin{lem}[Initial growth]
The evolutionary entropy satisfies $H'(1)>0.$
\label{eq:typeII_initial_growth}
\end{lem}

%
%
\begin{proof}

At \(x=1\), equations \eqref{eq:typeII_Zxi},
\eqref{eq:typeII_mxi}
and \
\eqref{eq:typeII_Txi} give
	\(Z(1)=1\), \(m(1)=0\), and \(T(1)=A\). Hence
	\(F(1)=-\beta\). Moreover,
	\begin{equation}
		\Phi(1)
		=
		\frac{\beta(1-q)}{\beta-1+q}.
		\label{eq:typeII_Phi1}
	\end{equation}
	Since \(q\in(0,1)\), we have \(\beta=-\log q>0\) and
	\(1-q>0\). It remains only to check the denominator. We have
	$
	\beta-1+q
	=
	-\log q-1+q.
	$
	But
	\[
	-\log q
	=
	\int_q^1 \frac{dt}{t}
	>
	\int_q^1 dt
	=
	1-q,
	\]
	and therefore \(\beta-1+q>0\). Thus both numerator and
	denominator in \eqref{eq:typeII_Phi1} are positive, and so
	\(\Phi(1)>0\).
	
	Consequently,
	\[
	\Phi(1)-F(1)=\Phi(1)+\beta>0.
	\]
	Since \(T(1)>0\) and \(T'(1)>0\), the differential identity
	\eqref{eq:typeII_Fprime} gives \(F'(1)>0\). Finally, because
	\(H'(x)=F'(x)\), we conclude that \(H'(1)>0\).
\end{proof}

We now analyse the asymptotic regime $x \to \infty$.

\begin{lem}[Asymptotic behaviour of entropy]
\label{lem:typeII_Hinfty}
As
\(
x\to\infty,
\)
\begin{equation*}
H(x)
\longrightarrow
H_\infty
=
-
\frac{
q\log q+(1-q)\log(1-q)
}{
A(1-q)+q
}.
\label{eq:typeII_Hinfty}
\end{equation*}
\end{lem}

\begin{proof}
From
\eqref{eq:typeII_Zxi},
\eqref{eq:typeII_mxi},
and
\eqref{eq:typeII_Txi}, respectively, 
it follows that
\begin{align*}
Z(x)\to Z_\infty &= \frac1{1-q}, \\
m(x)\to m_\infty &= \frac{q}{1-q}, \\
T(x)\to T_\infty &= A+\frac{q}{1-q}.
\end{align*}

\noindent Substituting these limits into
\eqref{eq:typeII_Hxi},
we obtain
\[
H_\infty
=
\frac{
\beta\frac{q}{1-q}
+\log\frac1{1-q}
}{
A+\frac{q}{1-q}
}.
\]

\noindent Recalling that
$
\beta=-\log q
$
and multiplying numerator and denominator by
\(
1-q
\)
proves the result.
\end{proof}

We next show that the asymptotic behaviour of the derivative of $H$ is determined by the limit $F_\infty$.

\begin{lem}[Asymptotic sign criterion]
\label{eq:typeII_negative_condition}
Let
$
F_\infty=\lim_{x\to\infty}F(x).
$
Then
\begin{equation*}
F_\infty<0
\quad\Longleftrightarrow\quad
q^A<1-q,
\end{equation*}
\begin{equation*}
F_\infty>0
\quad\Longleftrightarrow\quad
q^A>1-q.
\end{equation*}
\end{lem}

\begin{proof}
We first show that
\begin{equation}
F_\infty
=
\frac{
-\log(1-q)+A\log q
}{
A+\frac{q}{1-q}
}.
\label{eq:typeII_Finfty}
\end{equation}

Recalling from \eqref{eq:typeII_F} that
$
F(x)
=
\frac{
\log Z(x)-\beta A
}{
T(x)
}
$
and introducing in this expression the limits $Z_\infty$, $m_\infty$ and $T_\infty$ obtained in the proof of
Lemma~\ref{lem:typeII_Hinfty},
we obtain
\[
F_\infty
=
\frac{
-\log(1-q)-\beta A
}{
A+\frac{q}{1-q}
},
\]
whereupon substituting
$
\beta=-\log q
$
yields
\eqref{eq:typeII_Finfty}.

Since the denominator in
\eqref{eq:typeII_Finfty}
is positive, the sign of
\(
F_\infty
\)
is determined by the sign of the numerator 
\(
-\log(1-q)+A\log q.
\)
Thus
$
F_\infty<0
$
if and only if
$
A\log q<\log(1-q),
$
which is equivalent to
$
q^A<1-q.
$
Similarly,
$
F_\infty>0
$
is equivalent to
$
q^A>1-q,
$
completing the proof.
\end{proof}

We now provide a complete characterization of the transition between monotonic entropy growth and interior entropy maximization.

\begin{thm}[Critical threshold for entropy regimes]
\label{thm_boundary}

The following characterization holds.

\begin{enumerate}

\item
The equation
$
q^A = 1-q
$
has a unique solution
$
q_c(A)\in(0,1).
$

\item
If
$
q < q_c(A),
$
then
$
H(x)
$
is strictly increasing in
$
[1,\infty).
$

\item
If
$
q > q_c(A),
$
then
$
H(x)
$
has a unique interior maximum.

\end{enumerate}

\end{thm}

\begin{proof}
Define
\[
f_A(q)=q^A+q-1.
\]
Then $f_A(0)=-1$, $f_A(1)=1$, and
\[
f_A'(q)=Aq^{A-1}+1>0,
\]
so $f_A$ is strictly increasing on $(0,1)$ and has a unique zero $q_c(A)\in(0,1)$.

Assume first that $q<q_c(A)$, equivalently $q^A<1-q$. By Lemma~\ref{eq:typeII_negative_condition}, $F_\infty<0$. Since $\Phi(x)\to0$ as $x\to\infty$, the identity
\[
F'(x)=\frac{T'(x)}{T(x)}[\Phi(x)-F(x)]
\]
implies that $F'(x)>0$ for all sufficiently large $x$. Lemma~\ref{eq:typeII_initial_growth} gives $F'(1)>0$. If a critical point existed, Theorem~\ref{Unimodal} would make it a strict local maximum, so the derivative would change from positive to negative there. It could then become positive again at large $x$ only after another critical point, contradicting Theorem~\ref{Unimodal}. Hence no critical point exists and $H$ is strictly increasing on $[1,\infty)$.

Assume next that $q>q_c(A)$, equivalently $q^A>1-q$. By Lemma~\ref{eq:typeII_negative_condition}, $F_\infty>0$. Since $\Phi(x)\to0$, the same differential identity yields $F'(x)<0$ for all sufficiently large $x$, whereas Lemma~\ref{eq:typeII_initial_growth} gives $F'(1)>0$. By continuity, a critical point exists. Theorem~\ref{Unimodal} implies that it is unique and is a strict local maximum. This proves the characterization.
%
%
%
%
%
%
%
%
%
%
%
%
\end{proof}

Observe that the critical case
$
q=q_c(A)
$
corresponds precisely to the boundary case
$
F_\infty=0,
$
which separates the distinct asymptotic behaviours.

Theorem \ref{thm_boundary} shows that the geometric reproductive regime exhibits a sharp transition determined by the unique solution
$
q_c(A)
$
of the polynomial equation
\[
q^A+q-1=0.
\]

For
$
q<q_c(A),
$
the entropy grows monotonically with reproductive lifespan, whereas for
$
q>q_c(A)
$
the entropy possesses a unique finite maximizer.

The geometric family thus constitutes the critical exponential boundary separating monotone entropy growth from finite entropy maximization.

For the original integer-valued Leslie problem, the continuous maximizer supplied by Theorem~\ref{thm_boundary} is used only to identify the candidate integer endpoints. In the finite-maximum regime, one evaluates $H(D)$ at the adjacent integers around the continuous critical point and selects the larger value. In the monotone or asymptotic regime, no finite endpoint is selected by entropy maximization alone; the biologically meaningful summaries are then the reproductive quantiles $B_{50},B_{90},B_{95}$ and $B_{99}$ of the limiting reproductive distribution.


\subsection{On Type III survivorship profiles}

In principle, one may also consider survivorship functions with
strongly decreasing hazards during the juvenile phase, corresponding
to the classical Type III demographic regime.

However, from the perspective of reproductive entropy, the relevant
object is not the juvenile survivorship pattern itself but rather the
structure of the reproductive distribution after maturity.

Indeed, once the reproductive age is reached, many Type III
survivorship profiles become well approximated by either nearly
constant adult hazards or very low adult mortality. Consequently,
their reproductive distributions are often close to the geometric and
rectangular regimes analysed in the previous sections.

For this reason, and because the geometric regime already provides a
natural analytical description of the post-maturity dynamics of a wide
range of populations, a separate detailed treatment of Type III
survivorship functions is not required for the purposes of the present
work.


\FloatBarrier

\section{Validation of the theory}\label{sec_Validation}

The preceding sections developed a theory linking evolutionary
entropy, reproductive distributions, and the geometric structures
generated by open-group demographic models. We now test these
predictions against demographic matrices from a large collection of
animal species, examining observed reproductive windows,
entropy-optimal classes, 
distributional agreement, and the critical-threshold predictions of
Theorem~\ref{thm_boundary}. We also compare entropy-derived
quantities with independent life-history variables, including
phylogenetically corrected analyses. Complete species-level results
are reported in Tables~\ref{tab:SpeciesSummary1}
and~\ref{tab:SpeciesSummary2} of the Appendix.

\subsection{Data set and species selection}

Demographic matrices were extracted from the COMADRE Animal Matrix
Database~\cite{COMADRE}. Since the theory is formulated for
age-structured reproductive windows, we focused on matrices
representable as Leslie-type models, and retained only female
matrices with annual projection intervals to ensure comparability
across species.

The original database contained 3548 matrices, of which 1874 were
Leslie open-group, 652 pure Leslie, 844 of other types, 175 lacked
fertility information, and 3 had reproductive classes unreachable
from the initial stage. We further restricted attention to matrices
with $0.75 \leq \lambda \leq 1.25$, where $\lambda$ is the dominant
eigenvalue, excluding populations far from demographic equilibrium,
and discarded matrices with evolutionary entropy $H<0.01$, for which
reproductive contributions are concentrated in essentially a single
reproductive class.

Many species are represented in COMADRE by multiple matrices,
corresponding to different populations, years, or environmental
conditions. To avoid pseudoreplication, each species contributed a
single representative matrix, selected as the one with the largest
number of age classes; ties were broken by proximity of
$\lambda$ to unity.

Four species were excluded from the final analyses. One fossil taxon (Hystrix refossa) was removed because it is not comparable to extant populations. One redundant taxonomic entry (Chen caerulescens) was merged with Anser caerulescens. Two additional species (Dryocopus pileatus and Alouatta seniculus) were excluded because an unambiguous correspondence with the phylogenetic tree could not be established.

All entropy-based predictions were computed exclusively from the
demographic information contained in the matrices. Independent
life-history variables --- compiled from
AnAge~\cite{Tacutu2018}, Avibase~\cite{Lepage2024},
FishBase~\cite{FroesePauly2024}, EURING~\cite{EURING2024}, the Bird
Banding Laboratory~\cite{BBL2024}, and Animal Diversity
Web~\cite{ADW2024} --- were introduced only after the entropy
calculations had been completed, and serve solely as external
validation. Phylogenetically corrected analyses were performed on the
subset of 130 species for which an unambiguous correspondence with
the phylogenetic tree could be established. Matrix population models
are demographic reconstructions subject to sampling, stochasticity,
and estimation uncertainty; the analysis should therefore be read as
a large-scale test of whether entropy-based regularities emerge
across independently constructed matrices, rather than as a
species-by-species census. The restriction to animal annual
matrices is discussed further in subsection \ref{sub_limitations}.

The prediction protocol was fixed before comparisons with external life-history variables were made. For each matrix we first computed the reproductive masses, the cumulative reproductive distribution, the entropy curve, and the corresponding entropy-derived endpoints or quantiles. No observed longevity, age at last reproduction, mean reproductive age, or external estimate of generation time was used in selecting the entropy-optimal class. This separation is important: it makes the comparisons below genuine external biological validations of the entropy construction rather than fitted regressions to life-history data.

\subsection{Entropy-optimal reproductive classes}

The theory developed in sections \ref{sec_Leslie},
\ref{sec_exp_model} and
\ref{sec_entropy_maximization}
predicts that each reproductive window is associated with an
entropy-optimal class --- the point beyond which additional
reproduction yields only marginal gains in evolutionary entropy.

To test this, we constructed an entropy curve for each representative
matrix. Starting from the observed reproductive window, reproduction
was progressively extended using the demographic structure implied by
the corresponding Leslie matrix. For pure Leslie matrices, the observed
window was used directly. For open-group matrices, the extension was
performed using the terminal fertility and survival parameters already
present in the matrix, thereby preserving the demographic regime
observed at the end of the reproductive period.


The entropy-optimal class was defined as the earliest reproductive
extension for which entropy reached its maximum value or became
indistinguishable from its asymptotic limit. Biologically, this class
identifies an effective reproductive endpoint beyond which additional
reproduction contributes only marginally to evolutionary entropy.
Species may continue reproducing beyond the predicted class without
contradicting the theory, particularly in the geometric tails of
open-group populations described in \ref{sec_open_group}.

To compare species, we computed the ages $B_{50}, B_{90}, B_{95},
B_{99}$ at which 50\%, 90\%, 95\%, and 99\% of total reproductive
contribution has accumulated. Among these, $B_{50}$ --- the effective
centre of reproductive activity --- plays a central role in what
follows.

\begin{figure*}[ht]
\centering
\includegraphics[width=0.85\textwidth]{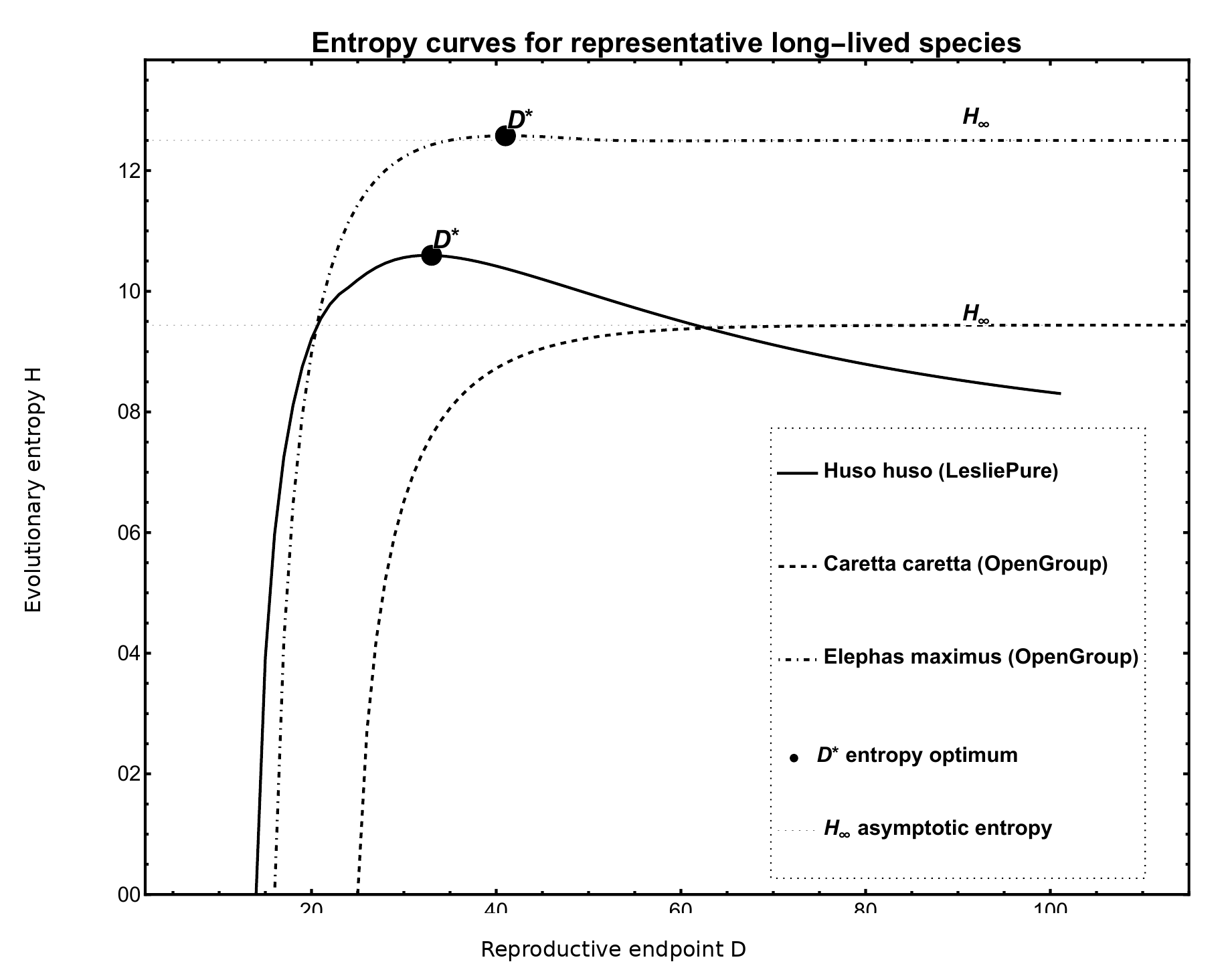}
\caption{
Entropy curves for three representative long-lived species. Filled
circles indicate entropy-optimal classes $D^{*}$ and horizontal dotted
lines indicate asymptotic entropy values $H_{\infty}$.
}
\label{fig:LongLivedEntropyCurves}
\end{figure*}

Figure~\ref{fig:LongLivedEntropyCurves} illustrates the diversity of
entropy trajectories observed among long-lived species. Although all
three species are characterized by extended lifespans and prolonged
reproductive periods, their entropy curves exhibit markedly different
behaviours.

The curve for \emph{Huso huso} reaches a well-defined finite maximum,
after which evolutionary entropy decreases as reproduction is extended.
This behaviour corresponds to the finite-optimum regime predicted by
the theory. In contrast, \emph{Caretta caretta} displays monotonic
convergence towards an asymptotic entropy value, illustrating the
asymptotic regime. The curve of \emph{Elephas maximus} occupies an
intermediate position, attaining a maximum that lies very close to the
corresponding asymptotic value.

These examples show that evolutionary entropy is not determined by
longevity alone: despite broadly similar life-history profiles, the
three species exhibit distinct entropy structures and therefore
different entropy-optimal reproductive classes.

Figure~\ref{fig:B50PercentileErrors} summarizes the distribution of
prediction errors for the reproductive median $B_{50}$ across all
species in the dataset. The concentration of points near the lower part
of the graph indicates that most species exhibit only small differences
between observed and predicted reproductive medians.

The results reveal a remarkable level of agreement. In approximately 61\% of species there is exact coincidence between observed and predicted values. Nearly 80\% of species are predicted within a single reproductive class, more than 85\% within two classes, and over 90\% within three classes. Only a small minority of species display larger deviations.

Both the individual examples in Figure~\ref{fig:LongLivedEntropyCurves} and the large-scale agreement illustrated in Figure~\ref{fig:B50PercentileErrors} suggest that entropy-optimal classes capture a genuine structural feature of reproductive organization across a wide range of animal taxa.

%
%



\subsection{Agreement between predicted and observed reproductive distributions}

To quantify agreement between theory and observations across the entire dataset, we examined two measures of similarity between distributions: the overlap coefficient \cite{Inman1989}
\begin{equation}
	\label{eq_overlap}
	O(p,q)
	=
	\sum_j \min(p_j,q_j),
\end{equation}
which measures the proportion of reproductive contribution shared
by the two distributions, and the Bhattacharyya coefficient
\cite{Pastore2019}
\begin{equation}
	\label{eq_Bhattacharyya}
	BC(p,q)
	=
	\sum_j \sqrt{p_j q_j},
\end{equation}
which measures overall distributional similarity. For both
coefficients, values close to $1$ indicate strong similarity.

For each species, the observed reproductive distribution was
calculated directly from the demographic matrix, whereas the predicted
distribution was obtained from the entropy-optimal class derived from
the theory. Distributional agreement was then measured using both the
overlap and Bhattacharyya coefficients.

The results are shown in
Figures~\ref{fig:OverlapDistribution} and
\ref{fig:BhattacharyyaDistribution}. In both figures, species are
ordered by percentile on the horizontal axis and the corresponding
similarity coefficient is shown on the vertical axis.

For the overlap coefficient, 50\% of species have similarity
0.95 or higher, 70\% have similarity 0.84 or higher, and only 20\%
have similarity below 0.60. For the Bhattacharyya coefficient the
results are even stronger: 50\% of species have similarity 0.97 or
higher, 70\% have similarity 0.92 or higher, and only 20\% have
similarity below 0.77. Moreover, fewer than 8\% of species exhibit
overlap coefficients below 0.50, and fewer than 5\% exhibit
Bhattacharyya coefficients below 0.50.

\begin{figure*}[htbp]
\centering
\includegraphics[width=0.70\textwidth]{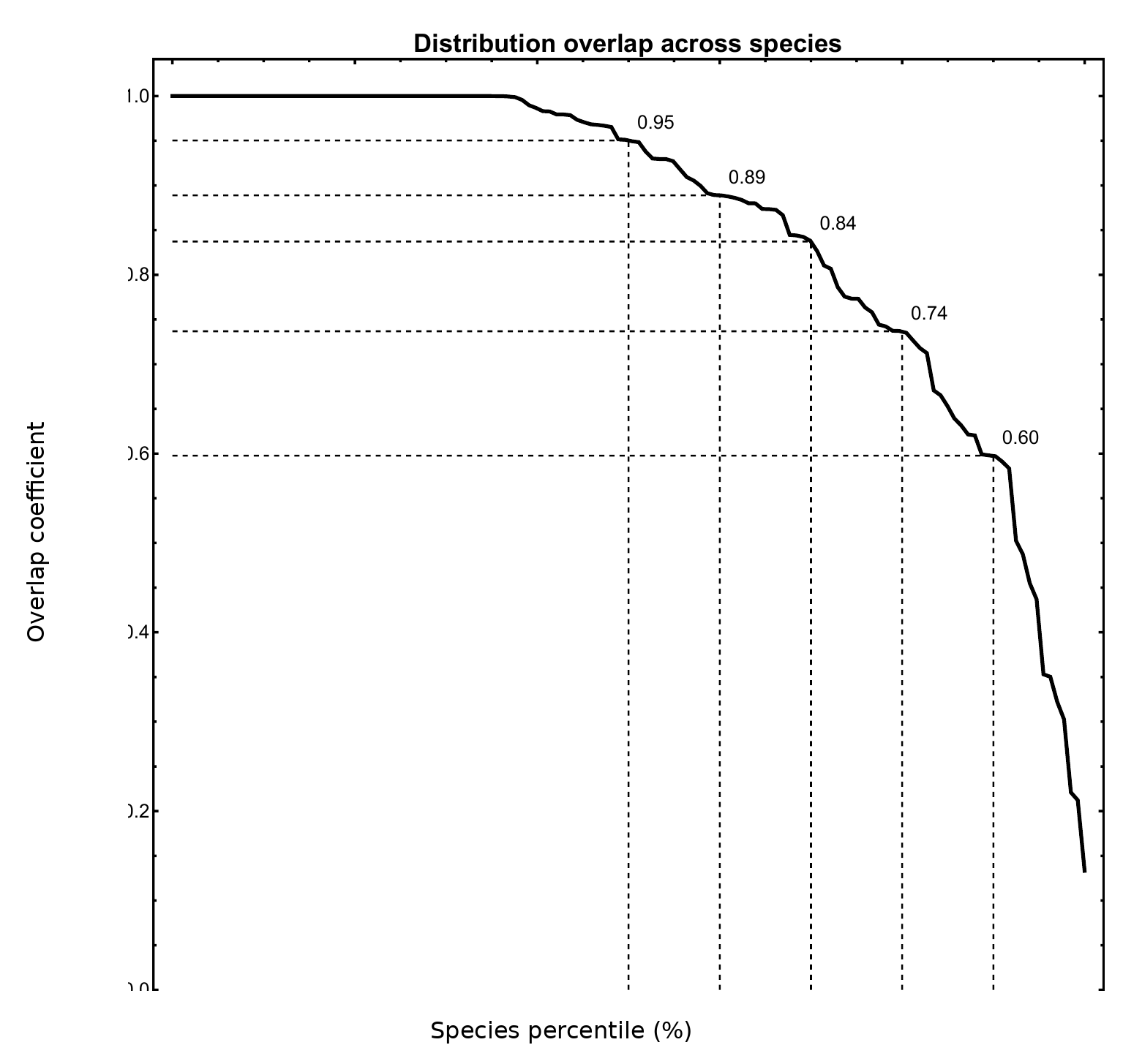}
\caption{Similarity between observed and
entropy-predicted reproductive distributions across all
analysed species measured by overlap coefficients.
}
\label{fig:OverlapDistribution}
\end{figure*}

\begin{figure*}[htbp]
\centering
\includegraphics[width=0.70\textwidth]{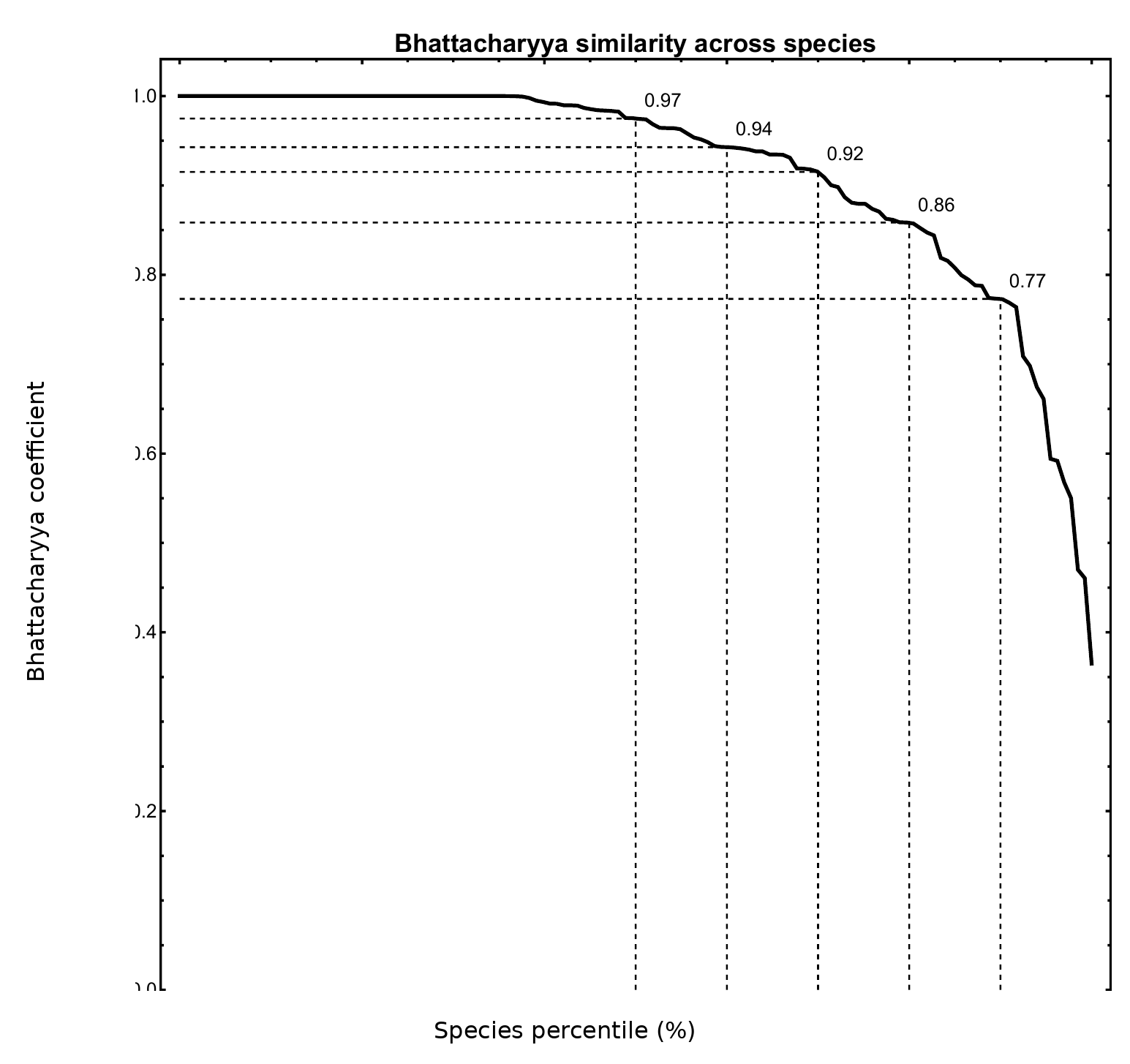}
\caption{
Similarity between observed and
entropy-predicted reproductive distributions across all
analysed species measured by Bhattacharyya coefficients.
}
\label{fig:BhattacharyyaDistribution}
\end{figure*}

Representative examples spanning a broad range of agreement levels are
shown in Figure~\ref{fig:DistributionExamples}. Whereas
\emph{Caretta caretta} exhibits virtually perfect agreement between
predicted and observed reproductive distributions,
\emph{Bonasa umbellus} and \emph{Elephas maximus} represent cases
of very high but imperfect correspondence. In contrast,
\emph{Huso huso} shows a substantial deviation between the
observed and entropy-predicted reproductive structures.

\begin{figure*}[htbp]
	\centering
	
	\begin{subfigure}[t]{0.48\textwidth}
		\centering
		\includegraphics[width=\textwidth]{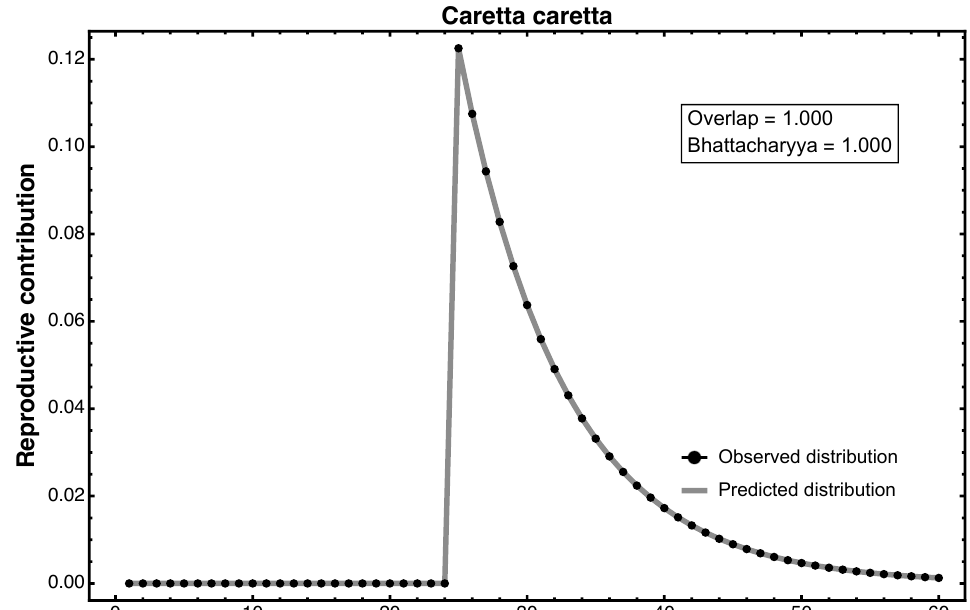}
		\caption{\emph{Caretta caretta}.}
	\end{subfigure}
	\hfill
	\begin{subfigure}[t]{0.48\textwidth}
		\centering
		\includegraphics[width=\textwidth]{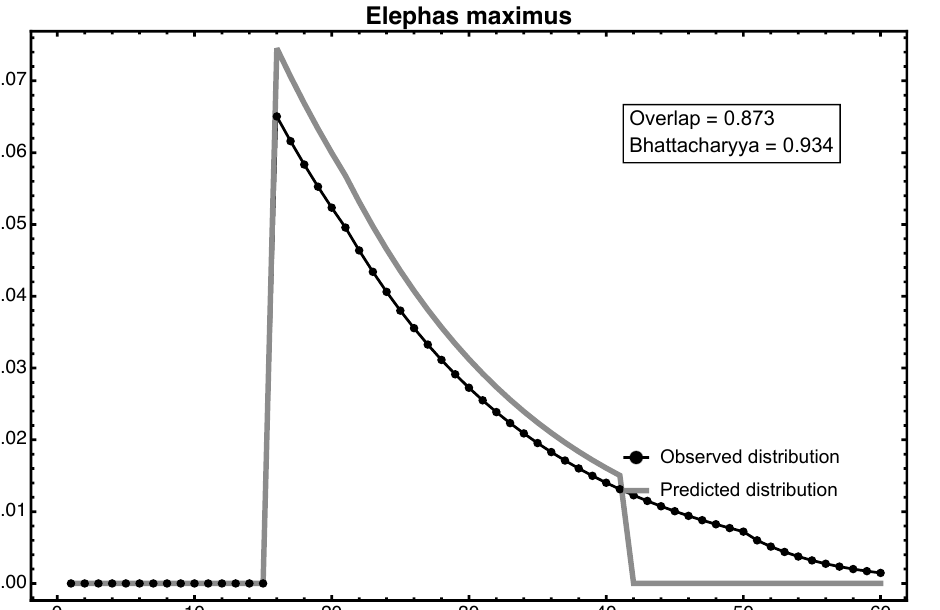}
		\caption{\emph{Elephas maximus}.}
	\end{subfigure}
	
	\vspace{0.5cm}
	
	\begin{subfigure}[t]{0.48\textwidth}
		\centering
		\includegraphics[width=\textwidth]{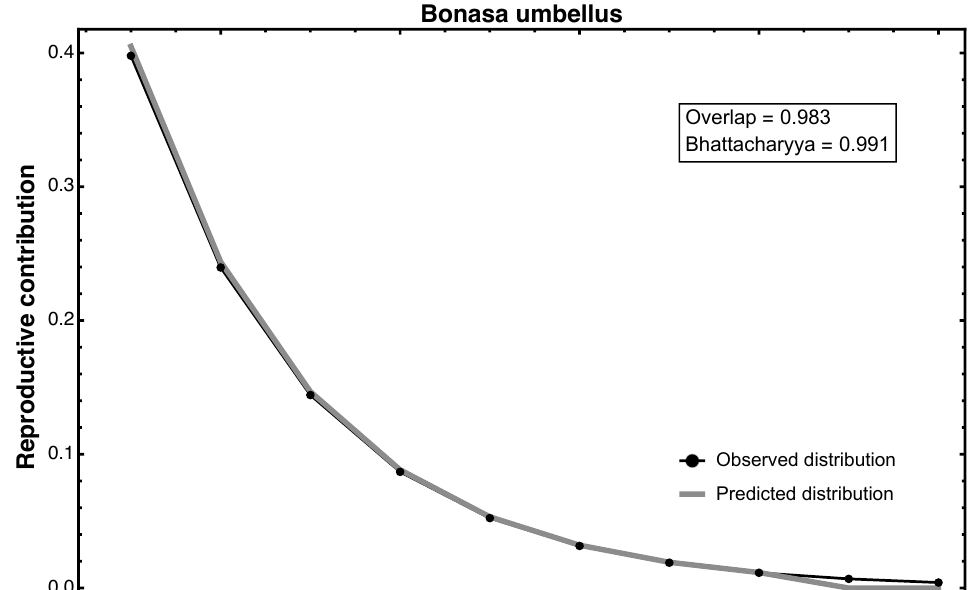}
		\caption{\emph{Bonasa umbellus}.}
	\end{subfigure}
	\hfill
	\begin{subfigure}[t]{0.48\textwidth}
		\centering
		\includegraphics[width=\textwidth]{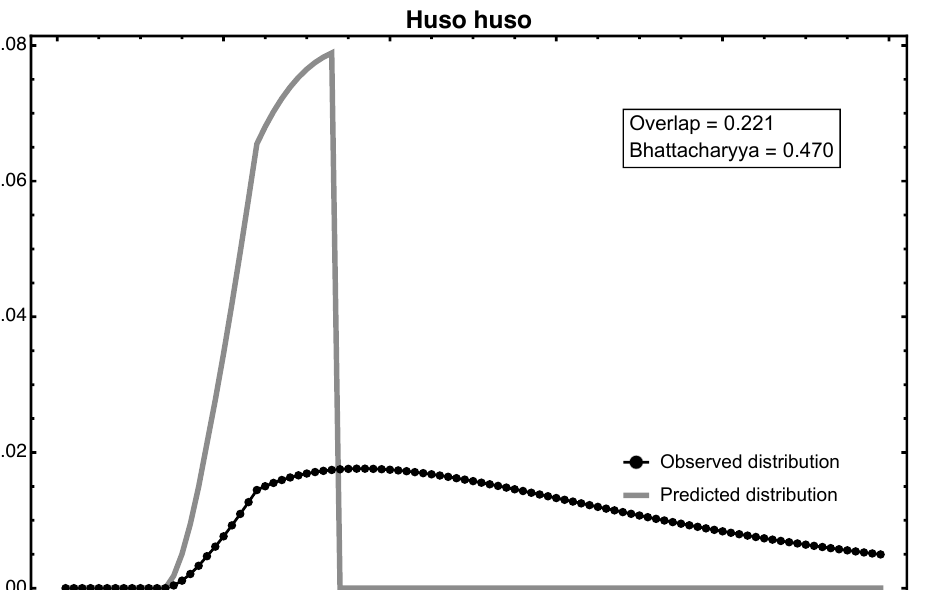}
		\caption{\emph{Huso huso}.}
	\end{subfigure}
	
	\caption{
		Observed and entropy-predicted reproductive distributions
		for four representative species spanning a broad range of agreement
		levels. Overlap and Bhattacharyya coefficient are reported within
		each panel.
	}
	\label{fig:DistributionExamples}
\end{figure*}

The overlap and Bhattacharyya analyses show that the vast majority of
species exhibit strong similarity between predicted and observed
reproductive distributions. This agreement is stronger than would be
expected from a comparison of isolated summary statistics alone. It
indicates that the entropy-optimal class retains information about the
shape of the reproductive contribution distribution, including the
location of its central mass and the rate at which reproductive
contributions decay across later age classes.

This point is important for the biological interpretation of the model.
An entropy-derived endpoint is not merely a fitted percentile: it is the
age class at which the full reproductive distribution, normalized by the
Euler--Lotka weights, gives the largest entropy per unit generation
time. Agreement at the distributional level therefore supports the view
that evolutionary entropy captures a structural feature of reproductive
organization rather than an accidental correlation with a single
life-history descriptor such as $B_{50}$.

These results provide distribution-level support for the
entropy-maximization principle. The extent to which entropy-derived
quantities predict independent life-history variables is examined in
the following sections.

Generation time, which depends on the entire reproductive
distribution rather than on a single reproductive class,
shows similarly strong agreement between predicted and observed
values (results not shown).

\subsection{Prediction of generation time}\label{subsec_generation_time_prediction}

Generation time depends on the entire reproductive contribution
distribution rather than on a single reproductive age class. It therefore
provides a stronger internal test of the theory than comparisons based
only on the reproductive median $B_{50}$. In the present data set, the
entropy-derived reproductive distribution gives close agreement with the
observed generation-time structure. Approximately $38\%$ of species show
exact coincidence between observed and predicted generation times, and
approximately $78\%$ are predicted with relative errors below $30\%$.

The importance of this comparison is that accurate prediction of
generation time requires more than locating the center of the
reproductive window. It requires the entropy-derived distribution to
capture the overall allocation of reproductive mass across the entire
post-maturity interval. Thus the generation-time analysis supports the
same conclusion as the overlap and Bhattacharyya comparisons: the
entropy-optimal class reflects the organization of the reproductive
schedule as a distribution, not simply the value of a single quantile.

\subsection{Prediction of entropy regimes}

The critical-threshold Theorem \ref{thm_boundary} provides a criterion
for predicting whether an open-group population exhibits a finite
entropy optimum or an asymptotic entropy regime.

Among all analysed open-group matrices, only three species were not
classified correctly:
\emph{Alces alces},
\emph{Canis lupus},
and \emph{Elephas maximus}
(Table~\ref{tab:CriticalTheoremExceptions}).
In each case, however, the hypotheses of Theorem
\ref{thm_boundary} are not satisfied because the finite survivorship
sequence crosses the critical threshold. All remaining open-group
species in the dataset are classified correctly.

\begin{table*}[htbp]
\centering
\caption{Species not classified correctly by the critical-threshold criterion of Theorem \ref{thm_boundary}.}
\label{tab:CriticalTheoremExceptions}
\begin{tabular}{llll}
\hline
Species & Observed & Predicted & Cause \\
\hline
\emph{Alces alces} &
Finite Maximum &
Asymptotic &
Threshold crossing \\

\emph{Canis lupus} &
Asymptotic &
Finite Maximum &
Threshold crossing \\

\emph{Elephas maximus} &
Finite Maximum &
Asymptotic &
Threshold crossing \\
\hline
\end{tabular}
\end{table*}

\subsection{Correlations with independent life-history variables}\label{subsec_phylogenetic}

The preceding analyses compared entropy-derived predictions with
quantities computed directly from the demographic matrices.
A substantially stronger test consists in comparing entropy-derived
quantities with independent life-history variables compiled from
external sources.

Because species share evolutionary history, comparative data cannot
generally be regarded as statistically independent
\cite{Felsenstein1985,Freckleton2002}. To account for this effect,
all correlation analyses were repeated using phylogenetically
corrected methods.

A phylogenetic tree for the analysed species was obtained from the
Open Tree of Life
\cite{Hinchliff2015,OpenTreeOfLife}. From this tree we constructed a
phylogenetic covariance matrix \(V\) describing the expected
covariance among species under a Brownian-motion model of trait
evolution \cite{Felsenstein1985}. Writing

\[
V=LL^{T},
\]
where \(L\) is the Cholesky factor of \(V\), trait vectors were
transformed according to

\[
x^{*}=L^{-1}x,
\qquad
y^{*}=L^{-1}y.
\]

Correlations and regressions performed in the transformed space are
equivalent to phylogenetic generalized least-squares (PGLS)
analyses \cite{Grafen1989,Freckleton2002}.

As a preliminary internal validation, we compared the predicted
reproductive median \(B_{50}^{\rm pred}\) with the observed
reproductive median \(B_{50}^{\rm obs}\) computed directly from the
demographic matrices. Since both quantities are derived from the
same reproductive windows, strong agreement is expected.
Indeed, after phylogenetic correction we obtain
\(r=0.937\),
\(\rho=0.963\),
and
\(R^2=0.878\)
(\(p<10^{-30}\)),
compared with an uncorrected correlation of
\(r=0.922\).
These results confirm the internal consistency of the
entropy-based predictions.

The more informative tests involve variables that are entirely
independent of the demographic matrices.

First, we compared
\(B_{50}^{\rm pred}\)
with the observed mean reproductive age,

\[
AgeMean
=
\frac{AgeFirst+AgeLast}{2},
\]
where \(AgeFirst\) and \(AgeLast\) denote the ages at first and last
reproduction compiled from independent life-history databases.

The association remains remarkably strong
(Figure~\ref{fig:B50AgeMeanPhylo}),
with phylogenetic Pearson correlation
\(r=0.884\),
phylogenetic Spearman correlation
\(\rho=0.875\),
and
\(R^2=0.782\)
(\(p<10^{-30}\)).
The corresponding uncorrected correlation was
\(r=0.825\).

\begin{figure*}[htbp]
\centering
\includegraphics[width=0.75\textwidth]{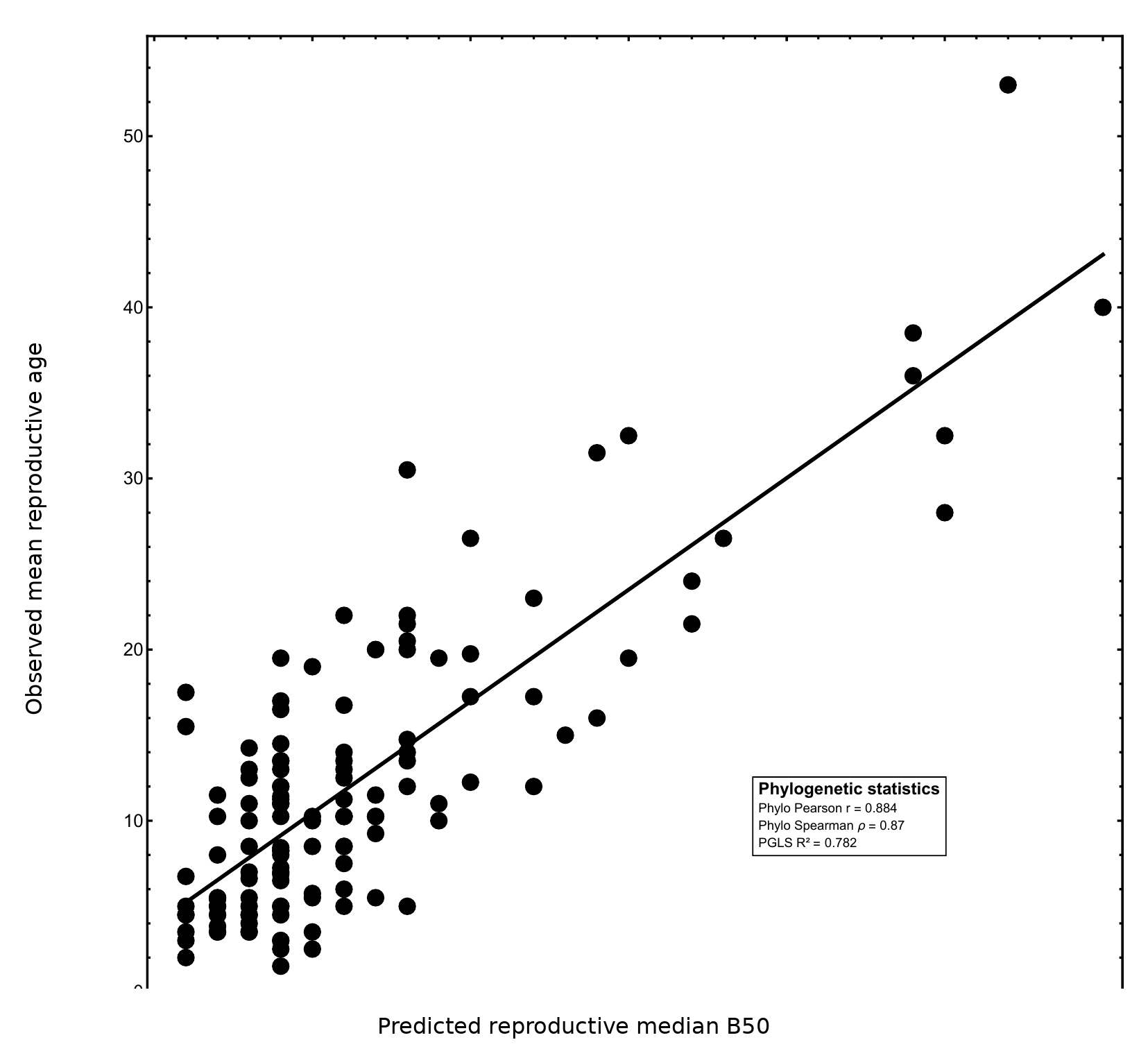}
\caption{
Phylogenetically corrected relationship between
\(B_{50}^{\rm pred}\)
and observed mean reproductive age.
}
\label{fig:B50AgeMeanPhylo}
\end{figure*}

Secondly, we compared the predicted generation time
\(T_{P}\)
with the same observed mean reproductive age. Since \(T_{P}\) depends on the
entire reproductive distribution rather than on a single
percentile, this provides an even broader test of the theory.

The results are shown in
Figure~\ref{fig:TAgeMeanPhylo}.
The phylogenetically corrected correlation is
\(r=0.889\),
with phylogenetic Spearman correlation
\(\rho=0.881\),
and
\(R^2=0.790\)
(\(p<10^{-30}\)).
The corresponding uncorrected correlation was
\(r=0.821\).

\begin{figure*}[htbp]
\centering
\includegraphics[width=0.75\textwidth]{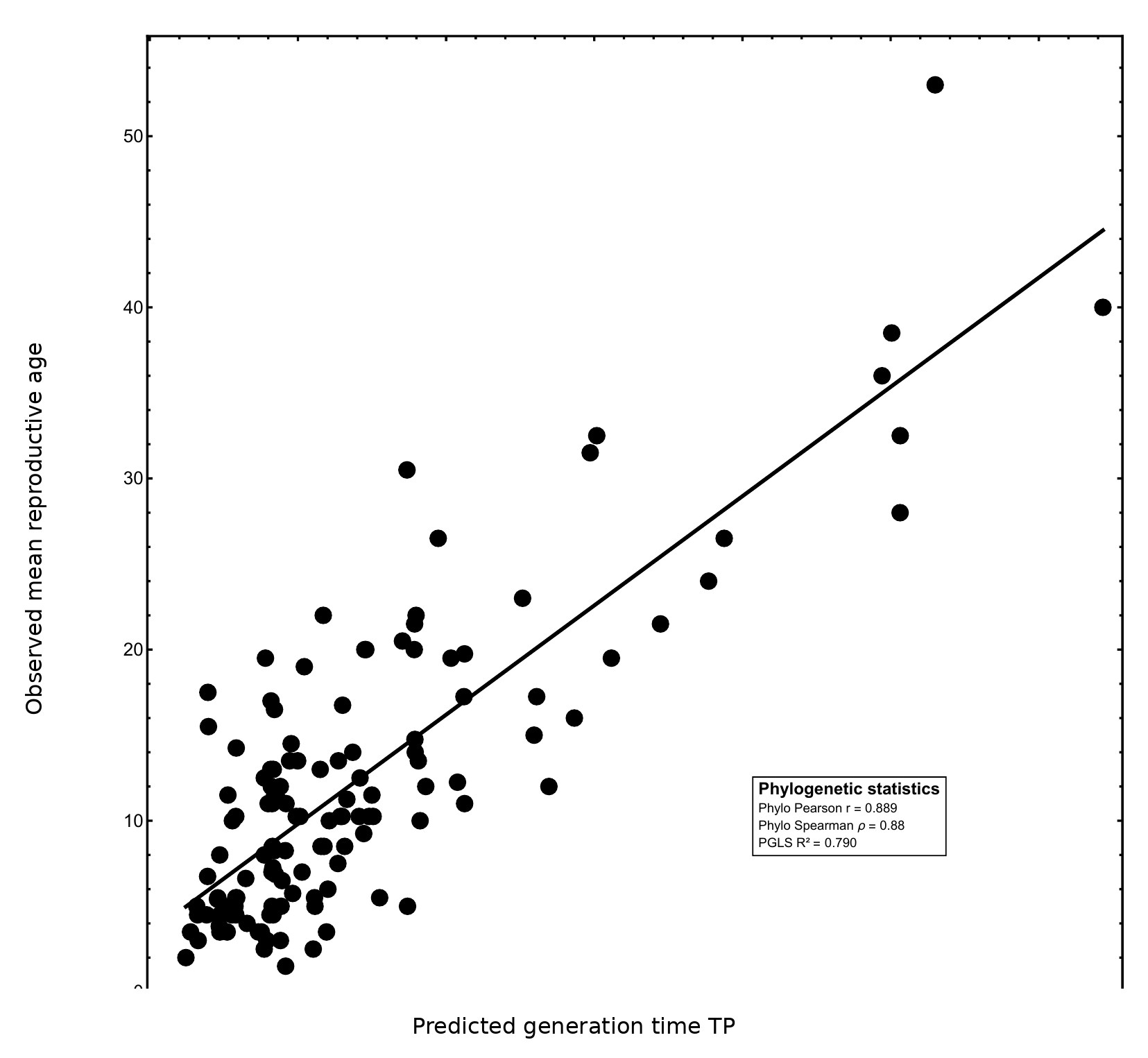}
\caption{
Phylogenetically corrected relationship between
\(T_{P}\)
and observed mean reproductive age.
}
\label{fig:TAgeMeanPhylo}
\end{figure*}

\begin{table*}[htbp]
\centering
\caption{
Summary of the principal phylogenetically corrected analyses
(\(N=130\)).
}
\label{tab:PGLSResults}
\begin{tabular}{lccccc}
\hline
Comparison &
Pearson &
Spearman &
\(R^2\) &
\(p\) &
Uncorrected \(r\) \\
\hline

\(B_{50}^{\rm pred}\) vs \(B_{50}^{\rm obs}\)
&
0.937
&
0.963
&
0.878
&
$<10^{-30}$
&
0.922
\\

\(B_{50}^{\rm pred}\) vs observed mean reproductive age
&
0.884
&
0.875
&
0.782
&
$<10^{-30}$
&
0.825
\\

\(T_{P}\) vs observed mean reproductive age
&
0.889
&
0.881
&
0.790
&
$<10^{-30}$
&
0.821
\\

\hline
\end{tabular}
\end{table*}

The close agreement between Pearson and Spearman coefficients
indicates that the observed relationships are not driven by a small
number of extreme observations and remain robust under rank-based
analysis. More importantly, the strongest results are obtained for
comparisons involving life-history variables compiled independently
of the demographic matrices. These associations remain extremely
strong after phylogenetic correction and are, if anything, slightly
stronger than their uncorrected counterparts.

These results show that the predictive power of the
entropy-maximization principle is not an artefact of shared
ancestry. The agreement between entropy-derived classes and
independent reproductive variables appears to reflect a genuine
biological regularity rather than taxonomic relatedness alone.

\subsection{Geometric reproductive organization}
\label{sec_geometric_R_O}

A notable feature of the demographic data analysed in this study is
the predominance of open-group Leslie matrices. Among the 2526 Leslie
matrices initially extracted from the COMADRE database, 1874
(74.2\%) correspond to open-group structures, whereas only 652
(25.8\%) correspond to finite Leslie populations. Thus, demographic
models incorporating persistent late-life reproduction are nearly
three times more common than purely finite reproductive windows in
the available empirical literature.

From the perspective of evolutionary entropy, this observation is of
particular relevance. As shown in Section~\ref{sec_open_group},
open-group populations generate reproductive distributions with
geometric tails, and it is within this setting that finite entropy
optima, asymptotic entropy regimes, and the critical-threshold
criterion naturally arise.

In this light, classical survivorship categories (Types I, II, and
III) appear to have limited explanatory value for reproductive
organization. Species with markedly different survivorship curves may
produce similar reproductive distributions and similar
entropy regimes. Reproductive distributions, particularly their
geometric structure, therefore provide a more direct description of
reproductive organization than survivorship curves alone.

\subsection{Robustness checks}
Two potential concerns regarding the preceding results deserve explicit analysis: whether predictive success is driven by the geometric tails of open-group matrices, and whether it merely reflects the well-known association between maturity and reproductive timing. In this section we address each in turn.

\subsubsection{Pure Leslie vs. Open Group Leslie}

A potential concern is that the strong agreement between entropy-derived predictions and observed reproductive windows might be driven primarily by the geometric reproductive tails naturally present in open-group Leslie populations. Since a large fraction of the analysed species belongs to this class, it is important to determine whether the predictive success persists in purely finite demographic systems.

To address this question, all analyses were repeated separately for Pure Leslie matrices and Open-Group Leslie matrices. The results are summarized in Table~\ref{tab:PureOpenComparison}.

\begin{table*}[t]
	\centering
	\caption{Comparison of predictive performance for Pure Leslie and Open-Group Leslie populations.}
	\label{tab:PureOpenComparison}
	\begin{tabular}{lcc}
		\hline
		Metric & Pure Leslie & Open Group \\
		\hline
		Species & 32 & 98 \\
		$r(B_{50},B_{50}^{\mathrm{Pred}})$ & 0.943 & 0.914 \\
		$\rho(B_{50},B_{50}^{\mathrm{Pred}})$ & 0.981 & 0.940 \\
		Mean absolute $B_{50}$ error & 1.47 & 1.47 \\
		Species within two classes (\%) & 90.6 & 83.7 \\
		$r(T,T_{\mathrm{Pred}})$ & 0.933 & 0.850 \\
		$\rho(T,T_{\mathrm{Pred}})$ & 0.984 & 0.877 \\
		Mean absolute $T$ error & 1.63 & 2.75 \\
		Mean overlap coefficient & 0.911 & 0.838 \\
		Mean Bhattacharyya coefficient & 0.947 & 0.907 \\
		\hline
	\end{tabular}
\end{table*}

The results reveal no deterioration of predictive performance in Pure Leslie populations. On the contrary, several metrics are consistently stronger in the finite Pure Leslie class than in Open-Group populations. Correlations between observed and predicted reproductive medians reach $r=0.943$ for Pure matrices, compared with $r=0.914$ for Open-Group matrices. Similar improvements are observed for generation times, overlap coefficients and Bhattacharyya coefficients.

Because Pure Leslie matrices lack the explicit geometric infinite tails characteristic of open-group systems, a structural dependency would imply a drop in predictive accuracy for finite matrices. However, we observe the opposite pattern.


The difference is especially pronounced for generation times. The correlation between observed and predicted generation times increases from $r=0.850$ in Open-Group populations to $r=0.933$ in Pure Leslie populations, while the mean absolute prediction error decreases from 2.75 to 1.63 age classes. Likewise, the overlap and Bhattacharyya coefficients remain remarkably high in the Pure Leslie subset, indicating that the agreement extends beyond summary statistics and persists at the level of the full reproductive distributions.

The largest discrepancies are concentrated in a very small number of extreme cases. In particular, the most substantial deviations occur in \textit{Huso huso}, one of the longest-lived species in the dataset and an extreme demographic outlier. Nevertheless, even with this species retained in the analysis, predictive performance remains strong in both matrix classes.

These results demonstrate that the empirical success of entropy-derived reproductive windows is not a structural artefact of the open-group formalism or its geometric tails. Because the predictive agreement persists,  and is in several respects even  stronger,  in finite pure Leslie populations, the entropy-derived model appears to capture a genuine demographic regularity independent of matrix representation.


\subsubsection{Comparison with a Maturity-Based Baseline}
\label{subsec_baseline}

Age at first reproduction is one of the most widely used life-history
descriptors and is naturally associated with subsequent reproductive
timing. To verify that the predictive power of entropy-derived reproductive windows does not merely reflect the well-known association between maturity and subsequent reproductive timing, we compared the model against a simple maturity-based baseline.

Table~\ref{tab:BaselineComparison} compares the predictive performance
of this baseline with that of the entropy-derived quantities.

\begin{table*}[t]
	\centering
	\caption{Comparison between a maturity-based baseline and the entropy-derived predictions.}
	\label{tab:BaselineComparison}
	\begin{tabular}{lcc}
		\hline
		Metric & Maturity baseline & Entropy model \\
		\hline
		$r(B_{50}^{\mathrm{Obs}},\cdot)$ & 0.843 & 0.922 \\
		Mean absolute error of $B_{50}$ & 4.68 & 1.47 \\
		$r(T^{\mathrm{Obs}},\cdot)$ & 0.809 & 0.872 \\
		Mean absolute error of $T$ & 6.10 & 2.47 \\
		\hline
	\end{tabular}
\end{table*}

As expected, age at first reproduction alone already provides a
reasonably informative predictor. Correlations between maturity and
the observed reproductive median and generation time reach
$r=0.843$ and $r=0.809$, respectively. Nevertheless, the
entropy-derived predictions substantially improve predictive accuracy.

For reproductive medians, the mean absolute prediction error decreases
from 4.68 to 1.47 age classes, corresponding to a reduction of
approximately $68.6\%$. For generation times, the mean absolute error
decreases from 6.10 to 2.47 age classes, corresponding to a reduction
of approximately $59.4\%$.

Thus, the predictive success of the entropy model extends beyond the demographic variance captured by maturity alone. While age at first reproduction establishes a baseline for reproductive timing, the entropy-derived windows provide a significantly sharper resolution of both reproductive medians and generation times, capturing deeper underlying demographic structures.

%

\subsection{Limitations and future directions}\label{sub_limitations}

Several limitations in this work should be acknowledged.

Firstly, the empirical analysis was restricted to animal populations
represented by annual projection matrices. Therefore, the generality of
the results outside this subset remains to be established. In
particular, demographic matrices with seasonal time steps, strongly
stage-structured models, and plant populations with dormancy or clonal
reproduction may require additional normalization before the present
entropy formulas can be used without modification.

Secondly, a single representative matrix was selected for each species.
Although this avoids over-representing species with numerous demographic
studies, it necessarily omits part of the demographic variability
associated with different populations, environments, and sampling
protocols. A natural next step is to repeat the analysis at the
population level, where several matrices for the same species could be
used to examine whether entropy-derived reproductive windows are stable
under ecological variation.

\begin{figure*}[ht]
\centering
\includegraphics[width=0.85\textwidth]{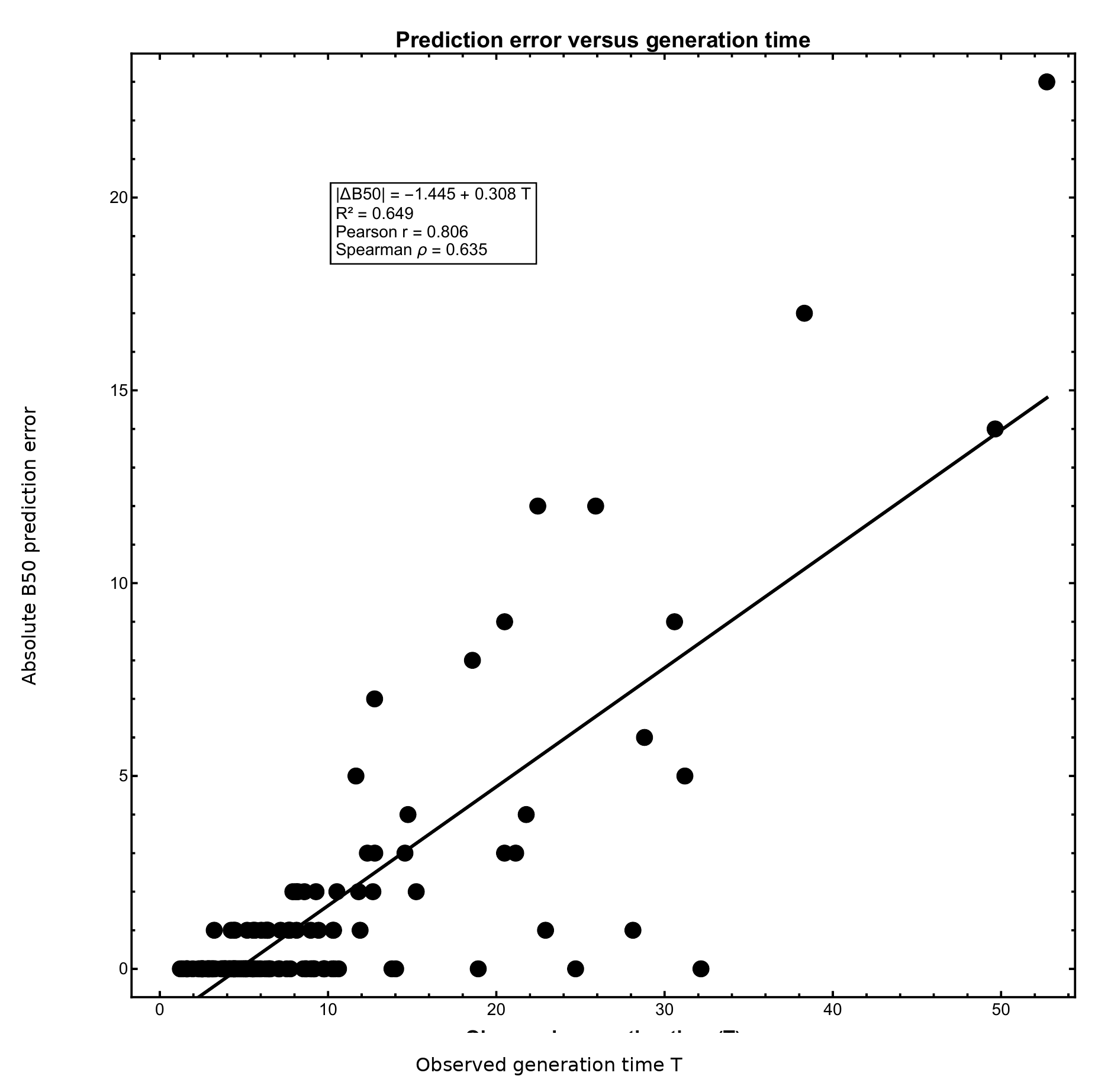}
\caption{
Absolute prediction error of the reproductive median $B_{50}$ as a
function of observed generation time $T$. Prediction error
increases strongly with generation time
($r=0.806$, $R^2=0.649$, $p<10^{-31}$).
}
\label{fig:B50ErrorT}
\end{figure*}

Thirdly, the theoretical analysis focuses primarily on Leslie
models and, in particular, on open-group populations with
geometric reproductive tails. These models admit exact analytical
treatment and lead naturally to the critical threshold established
in Theorem~\ref{thm_boundary}. Real populations may, however,
exhibit more complex demographic features, including variable
fertility windows and environmental stochasticity.

Finally, both demographic matrices and independent life-history
variables are subject to measurement uncertainty and sampling
variation.

Despite these limitations, the empirical agreement between theory
and observation is remarkably strong. Internal comparisons between
predicted and observed reproductive quantities exhibit very high
agreement, while the phylogenetically corrected analyses reported in
the previous subsection show that strong associations persist even
when entropy-derived quantities are compared with independent
life-history variables compiled from external sources.

Residual discrepancies are concentrated primarily among species with
long reproductive windows. Prediction error increases
significantly with generation time
(Figure~\ref{fig:B50ErrorT}),
indicating that the largest deviations occur in species exhibiting
extended late-life reproductive activity.

Finally, although the observed associations are consistent with the
entropy-maximization principle, they do not by themselves establish
that natural selection acts directly to maximize evolutionary
entropy. The present study documents a striking empirical
regularity. Clarifying the evolutionary mechanisms responsible for
this organization remains an important direction for future
research.

Future work should therefore proceed in three directions. First, the uncertainty of the demographic matrices should be propagated through the entropy calculations, producing confidence intervals for entropy-optimal classes and reproductive quantiles. Secondly, stage-structured and seasonal matrices should be treated directly rather than forced into an annual Leslie representation. Thirdly, stochastic environments should be incorporated through random products of projection matrices, allowing one to ask whether the entropy-derived reproductive horizon is stable under environmental variation. Each of these extensions would provide a sharper test of whether the regularity identified here is a robust property of reproductive organization or a feature of the deterministic annual matrices presently available.

\FloatBarrier

\section{Conclusions}

The present study identifies a simple mathematical relationship
between evolutionary entropy, reproductive organization, and
demographic life histories.

It also clarifies the biological position of evolutionary entropy relative to the demographic net reproductive number. The net reproductive number $R_0^{\rm dem}$ answers whether the lifetime schedule of survivorship and fertility produces replacement. The Perron factor $\lambda$ answers how fast the population grows or declines asymptotically. Evolutionary entropy answers a different question: how the growth-adjusted reproductive contribution is distributed through time.

The key point is that evolutionary entropy becomes biologically predictive only after the demographic problem is expressed in terms of reproductive masses. Once this is done, entropy maximization is no longer a purely formal calculation on survivorship curves. It becomes a statement about the probability distribution of realized reproductive contribution over age.

The theoretical analysis reveals two principal entropy regimes in
open-group populations, separated by a critical boundary case. Depending on the relationship between age at first
reproduction and the geometric parameter governing the reproductive
tail, evolutionary entropy either attains a finite maximum or
converges asymptotically towards a limiting value. An explicit
critical threshold separates these regimes and provides a criterion
for predicting the qualitative structure of reproductive windows
from demographic information.

The empirical analysis supports this theoretical framework. Across
130 species, entropy-derived quantities exhibit substantial agreement
with independently observed life-history variables. Approximately
61\% of species exhibit exact agreement between observed and predicted
reproductive medians, more than 90\% are predicted within three
reproductive classes, and strong correlations are obtained with age
at first reproduction, reproductive lifespan, life expectancy, and
maximum longevity.

These results are noteworthy because the predictions are derived
exclusively from demographic matrices. No information concerning
longevity, reproductive lifespan, reproductive termination,
generation time, or other external life-history variables enters the
entropy calculations. Nevertheless, many of these quantities are
closely associated with entropy-derived reproductive windows.

For a given Leslie system, the entropy-maximizing reproductive endpoint
necessarily lies within the reproductive sequence generated by the matrix,
since the optimization is performed over successive truncations of a finite
or open-ended sequence of reproductive classes. The substantive empirical result is therefore not this structural property itself, but the fact that entropy-derived reproductive classes remain strongly associated with independent life-history variables, suggesting that they capture biologically meaningful aspects of reproductive organization. This interpretation is further supported by analyses accounting for shared evolutionary history, which yield correlations ranging from $r=0.884$ to $r=0.937$ and coefficients of determination ranging from $R^2=0.782$ to $R^2=0.878$.

The theoretical and empirical results thus suggest that
evolutionary entropy captures an important aspect of reproductive
organization. A substantial fraction of reproductive structure
appears to be encoded in the distribution of reproductive
contributions and can be recovered from demographic information
alone. In this sense, reproductive distributions play a role
analogous to probability distributions in statistical physics:
they provide a compact description from which large-scale
regularities of reproductive organization can emerge.

The biological origin of the observed pattern remains an open
question. Given the uncertainty inherent in demographic
reconstructions
\cite{BuescuOliveiraSousa2023,Caswell,COMADRE,SalgueroGomez2016COMADRE}
and the extraordinary diversity of life histories and ageing
trajectories documented across the tree of life
\cite{Jones2014,Baudisch2008,deVries2023}, the extent of the observed
agreement is itself noteworthy. The results suggest that evolutionary
entropy captures a demographic regularity present across many
species, consistent with the broader interpretation of entropy as a
demographic descriptor proposed by Demetrius
\cite{Demetrius74,Demetrius75,Demetrius2004}. Across the iteroparous
species analysed here, reproductive windows tend to be organized near
entropy-optimal configurations, suggesting that the distribution of
reproductive contributions encodes important features of reproductive
organization.

From this perspective, the net reproductive number and evolutionary entropy play complementary roles in life-history interpretation. The net reproductive number is a scalar replacement total; entropy is a normalized temporal-distribution descriptor. A species with high lifetime reproductive output may still have low entropy if reproduction is concentrated in a narrow set of age classes, whereas a species with reproduction spread across many age classes may have a larger entropy even when its replacement output is not exceptional. This separation of amount from temporal organization is one reason the entropy framework is biologically informative for comparative demography.

The present work should therefore be understood as a mathematical and comparative description of a demographic regularity, rather than as a complete evolutionary explanation for that regularity. The results do not require the assumption that natural selection directly maximizes evolutionary entropy in every population. They show instead that the entropy functional selects reproductive age classes that are strongly aligned with independently observed reproductive schedules. This distinction is important. Establishing a robust quantitative pattern is a necessary step before one can determine whether the pattern is produced by selection, physiological constraint, ecological filtering, historical contingency, or some combination of these mechanisms. In this sense, the contribution of the present article is to identify a tractable mathematical structure that can be tested, generalized, and eventually embedded in broader evolutionary models of life-history organization.

\section*{Competing interests}

The authors declare that they have no competing interests.

\section*{Data availability}
All demographic matrices were obtained from the COMADRE Animal Matrix Database. Species-level results and derived quantities are reported in Appendix A. 

\section*{Code availability}
All analyses were performed in Wolfram Mathematica. The computational codes used to generate the results reported in this study are available from the corresponding author upon reasonable request.

\section*{Funding} J. B. was partially support by Fundação para a Ciência e a Tecnologia, UID/04561/2025, DOI: 10.54499/UID/04561/2025.
H. M. O. was partially supported by Fundação para a Ciência e a Tecnologia (FCT), Portugal, through CAMGSD, Instituto Superior Técnico, under projects UIDB/04459/2020, UIDP/04459/2020 and UID/PRR/04459/2025.

\section*{Author contributions}
H.M.O. developed the mathematical theory, designed the computational analyses and performed the numerical experiments. J.B. contributed to the mathematical theory, theoretical development, interpretation of the results and critical revision of the manuscript. S.N.E. contributed to the mathematical formulation and critical revision of the manuscript. All authors reviewed and approved the final version.

\appendix
\section{Summary variables and ordering of species}\label{app:species_results}

Tables~\ref{tab:SpeciesSummary1} and~\ref{tab:SpeciesSummary2}
summarise the principal quantities used throughout the paper to compare
observed reproductive schedules with the entropy-based predictions.

Species are ordered by increasing observed generation time
\(T\), providing a demographic gradient from species with highly
concentrated reproductive contributions to species exhibiting
progressively more extended reproductive windows.

\textit{Species} and \textit{ID} denote the species name and COMADRE
identifier. \textit{Type} indicates whether the matrix is a finite
Leslie matrix (\textit{Pure}) or an open-group Leslie matrix
(\textit{Open}). \textit{Reg} denotes the entropy regime:
\textit{Fin} (finite optimum), \textit{Asym} (asymptotic regime), or
\textit{Bound} (boundary case). The parameter \(\lambda\) is the
dominant eigenvalue of the selected matrix.

\(A\) denotes the first reproductive class and \(K\) the
entropy-optimal reproductive window. The ratio \(H/H_{\max}\)
measures the proximity of the observed reproductive structure to the
entropy-maximising regime.

\(B_{50}\) and \(B_{50P}\) are the observed and predicted
reproductive medians, respectively, and
\(\Delta B_{50}=B_{50P}-B_{50}\).
\(T\) and \(T_P\) are the corresponding observed and predicted generation times. \textit{Overlap} denotes the overlap coefficient
between observed and predicted reproductive 
distributions.

\textit{Mat}, \textit{LastRep}, and \textit{Long} denote,
respectively, age at maturity, last reported reproductive age, and
reported longevity obtained from independent life-history sources
including AnAge \cite{Tacutu2018}, Avibase \cite{Lepage2024},
FishBase \cite{FroesePauly2024}, EURING \cite{EURING2024}, the Bird
Banding Laboratory \cite{BBL2024}, and Animal Diversity Web
\cite{ADW2024}.

\setlength{\tabcolsep}{1.5pt}
\begin{table}[ht]
\centering
\scriptsize
\caption{Species summary table (Part 1).}\label{tab:SpeciesSummary1}
\resizebox{\textwidth}{!}{
\begin{tabular}{lllllllllllllllll}
\hline
Species & ID & Type & Reg & $\lambda$ & A & K &
$H/H_{\max}$ & $B_{50}$ & $B_{50P}$ &
$\Delta B_{50}$ & $T$ & $T_P$ &
Overlap & Mat & LastRep & Long \\
\hline

Opsopoeodus emiliae & 248178 & Pure & Fin & 0.985 & 1 & 3 & 1.000 & 1 & 1 & 0 & 1.221 & 1.221 & 1.000 & 1.000 & 3.000 & 4.000 \\
Callospermophilus lateralis & 249177 & Pure & Bound & 1.008 & 1 & 6 & 1.000 & 1 & 1 & 0 & 1.377 & 1.377 & 1.000 & 1.000 & 6.000 & 9.000 \\
Upupa epops & 248817 & Open & Asym & 1.038 & 1 & 26 & 1.000 & 1 & 1 & 0 & 1.587 & 1.587 & 1.000 & 1.000 & 9.000 & 11.10 \\
Jynx torquilla & 248657 & Open & Asym & 1.006 & 1 & 27 & 1.000 & 1 & 1 & 0 & 1.615 & 1.615 & 1.000 & 1.000 & 8.000 & 10.00 \\
Urocitellus armatus & 249847 & Pure & Fin & 0.970 & 1 & 5 & 1.000 & 1 & 1 & 0 & 1.635 & 1.635 & 1.000 & 1.000 & 5.000 & 7.000 \\
Lagopus leucura & 248690 & Open & Fin & 0.987 & 1 & 15 & 1.000 & 1 & 1 & 0 & 1.955 & 1.953 & 1.000 & 1.490 & 12.00 & 15.00 \\
Lagopus muta & 248700 & Open & Asym & 1.115 & 1 & 31 & 1.000 & 2 & 2 & 0 & 2.296 & 2.296 & 1.000 & 0.990 & 10.00 & 11.70 \\
Alcedo atthis & 248387 & Open & Asym & 1.052 & 2 & 21 & 1.000 & 2 & 2 & 0 & 2.363 & 2.363 & 1.000 & 1.000 & 15.00 & 21.00 \\
Poecile atricapillus & 248728 & Open & Asym & 1.177 & 2 & 25 & 1.000 & 2 & 2 & 0 & 2.503 & 2.503 & 1.000 & 1.000 & 9.000 & 12.00 \\
Vermivora chrysoptera & 248820 & Open & Fin & 0.794 & 1 & 8 & 0.997 & 2 & 2 & 0 & 2.505 & 2.367 & 0.983 & 1.000 & 6.000 & 8.000 \\
Bonasa umbellus & 248437 & Open & Fin & 0.945 & 1 & 8 & 0.996 & 2 & 2 & 0 & 2.513 & 2.373 & 0.983 & 1.000 & 8.000 & 11.00 \\
Puffinus tenuirostris & 248737 & Open & Fin & 0.999 & 1 & 16 & 0.963 & 1 & 1 & 0 & 2.527 & 1.963 & 0.971 & 5.000 & 30.00 & 38.00 \\
Campylorhynchus brunneicapillus\\ sandiegensis & 248455 & Open & Fin & 1.021 & 1 & 7 & 0.995 & 2 & 2 & 0 & 2.575 & 2.347 & 0.968 & 0.670 & 7.000 & 7.300 \\
Zingel asper & 248220 & Open & Asym & 1.151 & 2 & 32 & 1.000 & 2 & 2 & 0 & 2.768 & 2.768 & 1.000 & 3.000 & 6.000 & 8.000 \\
Dendragapus obscurus & 248511 & Open & Asym & 1.198 & 2 & 35 & 1.000 & 2 & 2 & 0 & 2.878 & 2.878 & 1.000 & 1.000 & 9.000 & 12.00 \\
Tamiasciurus hudsonicus & 249871 & Open & Asym & 1.006 & 1 & 24 & 1.000 & 3 & 3 & 0 & 2.903 & 2.903 & 1.000 & 1.000 & 8.000 & 9.800 \\
Zenaida macroura & 248826 & Open & Asym & 1.095 & 2 & 35 & 1.000 & 2 & 2 & 0 & 2.904 & 2.904 & 1.000 & 0.500 & 20.00 & 31.00 \\
Hirundo rustica & 248654 & Open & Asym & 0.984 & 2 & 36 & 1.000 & 2 & 2 & 0 & 2.937 & 2.937 & 1.000 & 1.000 & 10.00 & 15.00 \\
Agelaius phoeniceus & 248386 & Open & Fin & 1.000 & 1 & 8 & 0.997 & 3 & 3 & 0 & 3.106 & 2.902 & 0.973 & 1.000 & 10.00 & 15.80 \\
Vulpes vulpes & 249918 & Open & Fin & 1.001 & 1 & 6 & 0.935 & 2 & 2 & 0 & 3.134 & 2.285 & 0.889 & 0.830 & 10.00 & 15.00 \\
Anas platyrhynchos & 248401 & Open & Fin & 0.946 & 1 & 6 & 0.982 & 3 & 2 & -1 & 3.233 & 2.638 & 0.905 & 1.000 & 22.00 & 29.10 \\
Sciurus niger cinereus & 249824 & Pure & Bound & 1.216 & 2 & 9 & 1.000 & 3 & 3 & 0 & 3.241 & 3.241 & 1.000 & 1.250 & 12.00 & 15.00 \\
Ammodramus savannarum & 248393 & Open & Asym & 0.955 & 2 & 44 & 1.000 & 3 & 3 & 0 & 3.284 & 3.284 & 1.000 & 1.000 & 7.000 & 10.00 \\
Falco naumanni & 248519 & Open & Fin & 0.988 & 1 & 6 & 0.976 & 3 & 3 & 0 & 3.665 & 2.866 & 0.873 & 1.000 & 9.000 & 10.90 \\
Assa darlingtoni & 248232 & Open & Asym & 1.000 & 3 & 31 & 1.000 & 3 & 3 & 0 & 3.667 & 3.667 & 1.000 & 1.000 & 6.000 & 7.000 \\
Abax parallelepipedus & 248967 & Open & Asym & 1.037 & 3 & 33 & 1.000 & 3 & 3 & 0 & 3.766 & 3.766 & 1.000 & 2.000 & 5.000 & 6.000 \\
Nephtys incisa & 249983 & Pure & Bound & 0.760 & 2 & 5 & 1.000 & 4 & 4 & 0 & 3.863 & 3.863 & 1.000 & 1.000 & 4.000 & 5.000 \\
Canis lupus & 249186 & Open & Asym & 1.007 & 3 & 7 & 1.000 & 4 & 4 & 0 & 3.867 & 3.867 & 1.000 & 2.000 & 14.00 & 20.00 \\
Erythrocebus patas & 240418 & Open & Fin & 0.921 & 2 & 23 & 1.000 & 3 & 3 & 0 & 3.867 & 3.866 & 1.000 & 3.000 & 22.00 & 30.00 \\
Spermophilus dauricus & 249865 & Pure & Bound & 0.924 & 2 & 7 & 1.000 & 4 & 4 & 0 & 3.950 & 3.950 & 1.000 & 1.000 & 5.000 & 7.000 \\
Sturnella neglecta & 248783 & Open & Fin & 0.780 & 2 & 15 & 0.999 & 3 & 3 & 0 & 4.118 & 4.056 & 0.996 & 1.000 & 8.000 & 10.00 \\
Zoarces viviparus & 248223 & Pure & Fin & 1.110 & 3 & 7 & 1.000 & 4 & 4 & 0 & 4.126 & 4.126 & 1.000 & 2.000 & 8.000 & 10.00 \\
Puma concolor & 240521 & Pure & Fin & 1.072 & 3 & 13 & 1.000 & 4 & 4 & 0 & 4.207 & 4.207 & 1.000 & 2.500 & 14.00 & 18.00 \\
Centrocercus urophasianus & 248472 & Open & Fin & 0.979 & 1 & 5 & 0.909 & 3 & 2 & -1 & 4.240 & 2.613 & 0.744 & 1.000 & 6.000 & 9.000 \\
Accipiter cooperii & 248384 & Open & Fin & 0.932 & 2 & 13 & 0.998 & 3 & 3 & 0 & 4.307 & 4.146 & 0.987 & 2.000 & 15.00 & 20.00 \\
Buteo lineatus & 248445 & Open & Fin & 0.996 & 2 & 12 & 0.997 & 3 & 3 & 0 & 4.362 & 4.131 & 0.979 & 2.000 & 20.00 & 26.00 \\
Forpus passerinus & 248553 & Open & Fin & 1.028 & 1 & 5 & 0.788 & 2 & 1 & -1 & 4.403 & 1.921 & 0.726 & 1.000 & 8.000 & 10.50 \\
Cottus sp. & 248104 & Pure & Fin & 0.918 & 2 & 8 & 1.000 & 4 & 4 & 0 & 4.410 & 4.410 & 1.000 & 1.000 & 5.000 & 7.000 \\
Urocitellus columbianus & 249863 & Pure & Fin & 1.090 & 3 & 7 & 1.000 & 4 & 4 & 0 & 4.432 & 4.432 & 1.000 & 2.000 & 8.000 & 11.00 \\
Eidolon helvum & 249272 & Open & Fin & 1.157 & 2 & 11 & 0.994 & 3 & 3 & 0 & 4.446 & 4.083 & 0.965 & 2.000 & 24.00 & 31.00 \\
Marmota flaviventris & 249461 & Pure & Bound & 0.991 & 2 & 10 & 1.000 & 4 & 4 & 0 & 4.464 & 4.464 & 1.000 & 2.000 & 11.00 & 15.00 \\
Circus Circus maillardi & 253167 & Open & Fin & 0.998 & 2 & 11 & 0.995 & 4 & 3 & -1 & 4.467 & 4.123 & 0.967 & 2.000 & 12.00 & 15.00 \\
Scolytus ventralis & 249053 & Open & Asym & 1.239 & 4 & 30 & 1.000 & 4 & 4 & 0 & 4.585 & 4.585 & 1.000 & 1.000 & 2.000 & 3.000 \\
Buteo jamaicensis & 248444 & Open & Fin & 1.089 & 2 & 10 & 0.987 & 4 & 4 & 0 & 4.766 & 4.169 & 0.938 & 2.000 & 24.00 & 30.00 \\
Oncorhynchus tshawytscha & 248168 & Pure & Bound & 0.970 & 4 & 5 & 1.000 & 5 & 5 & 0 & 4.822 & 4.822 & 1.000 & 3.500 & 8.000 & 9.000 \\
Lycaon pictus & 249405 & Open & Fin & 0.980 & 2 & 10 & 0.982 & 4 & 4 & 0 & 4.917 & 4.235 & 0.930 & 1.750 & 12.00 & 16.00 \\
Anser anser & 248405 & Open & Asym & 1.177 & 3 & 60 & 1.000 & 4 & 4 & 0 & 4.992 & 4.992 & 1.000 & 3.000 & 24.00 & 31.00 \\
Anser Anser caerulescens & 253147 & Open & Fin & 1.167 & 2 & 13 & 0.997 & 4 & 4 & 0 & 5.014 & 4.772 & 0.979 & 3.000 & 26.00 & 33.00 \\
Perisoreus canadensis & 248714 & Open & Fin & 1.082 & 2 & 9 & 0.974 & 4 & 4 & 0 & 5.130 & 4.153 & 0.891 & 1.500 & 13.00 & 17.00 \\
Picoides borealis & 248726 & Open & Fin & 0.971 & 2 & 31 & 1.000 & 4 & 4 & 0 & 5.140 & 5.139 & 1.000 & 1.000 & 13.00 & 16.00 \\
Tamiasciurus douglasii & 249870 & Open & Fin & 0.790 & 2 & 9 & 0.973 & 4 & 4 & 0 & 5.160 & 4.162 & 0.889 & 1.000 & 8.000 & 9.000 \\
Milvus migrans & 248702 & Open & Fin & 1.049 & 1 & 5 & 0.860 & 4 & 3 & -1 & 5.191 & 2.789 & 0.620 & 2.000 & 18.00 & 24.00 \\
Accipiter gentilis & 248385 & Open & Fin & 1.070 & 2 & 10 & 0.978 & 4 & 4 & 0 & 5.463 & 4.599 & 0.909 & 2.000 & 20.00 & 27.00 \\
Notropis photogenis & 248156 & Open & Asym & 1.008 & 3 & 29 & 1.000 & 5 & 5 & 0 & 5.517 & 5.517 & 1.000 & 1.000 & 4.000 & 5.000 \\
Alces alces & 249121 & Open & Fin & 1.004 & 2 & 7 & 0.980 & 5 & 4 & -1 & 5.545 & 4.323 & 0.737 & 2.500 & 20.00 & 27.00 \\
Oncorhynchus clarkii lewisi & 248161 & Open & Asym & 1.008 & 4 & 37 & 1.000 & 5 & 5 & 0 & 5.561 & 5.561 & 1.000 & 3.000 & 8.000 & 12.00 \\
Tautogolabrus adspersus & 248219 & Pure & Bound & 1.000 & 2 & 8 & 1.000 & 6 & 6 & 0 & 5.576 & 5.576 & 1.000 & 2.000 & 8.000 & 10.00 \\
Lontra canadensis & 249401 & Open & Fin & 1.015 & 2 & 11 & 0.985 & 5 & 4 & -1 & 5.657 & 4.943 & 0.929 & 2.500 & 18.00 & 24.00 \\
Ciconia ciconia & 248501 & Open & Fin & 1.056 & 2 & 8 & 0.947 & 4 & 4 & 0 & 5.730 & 4.092 & 0.810 & 4.000 & 30.00 & 39.00 \\
Gavia immer & 248615 & Open & Fin & 1.010 & 1 & 8 & 0.656 & 1 & 1 & 0 & 5.889 & 1.981 & 0.773 & 6.000 & 25.00 & 31.00 \\
Rutilus rutilus & 248204 & Pure & Bound & 0.993 & 4 & 10 & 1.000 & 6 & 6 & 0 & 6.012 & 6.012 & 1.000 & 2.000 & 10.00 & 14.00 \\
Bubo virginianus & 248442 & Open & Fin & 1.056 & 2 & 8 & 0.926 & 5 & 4 & -1 & 6.030 & 4.102 & 0.786 & 2.000 & 22.00 & 29.00 \\
Phalacrocorax carbo & 253149 & Open & Fin & 1.185 & 2 & 10 & 0.981 & 5 & 5 & 0 & 6.283 & 5.216 & 0.880 & 3.000 & 35.00 & 43.00 \\
Larus ridibundus & 253150 & Open & Fin & 1.138 & 2 & 9 & 0.954 & 5 & 4 & -1 & 6.291 & 4.719 & 0.826 & 2.000 & 25.00 & 32.00 \\
Chlorocebus aethiops & 240395 & Open & Fin & 1.093 & 2 & 8 & 0.913 & 5 & 4 & -1 & 6.374 & 4.204 & 0.763 & 2.830 & 20.00 & 24.00 \\
\hline
\end{tabular}
}
\end{table}

\setlength{\tabcolsep}{1.5pt}
\begin{table}[ht]
\centering
\scriptsize
\caption{Species summary table (Part 2).}\label{tab:SpeciesSummary2}
\resizebox{\textwidth}{!}{
\begin{tabular}{lllllllllllllllll}
\hline
Species & ID & Type & Reg & $\lambda$ & A & K & $H/H\_{max}$ & $B\_{50}$ & $B\_{50P}$ & $\Delta B\_{50}$ & $T$ & $T_P$ & Overlap & Mat & LastRep & Long \\ 
\hline
Gulo gulo & 249375 & Open & Fin & 0.988 & 2 & 8 & 0.954 & 5 & 4 & -1 & 6.405 & 4.573 & 0.758 & 2.500 & 14.00 & 18.00 \\
Sternula antillarum browni & 248754 & Pure & Fin & 1.053 & 3 & 12 & 0.996 & 6 & 5 & -1 & 6.428 & 6.059 & 0.949 & 2.000 & 18.00 & 24.00 \\
Hemitragus jemlahicus & 249378 & Pure & Fin & 0.999 & 2 & 13 & 1.000 & 6 & 6 & 0 & 6.434 & 6.434 & 1.000 & 2.500 & 18.00 & 22.00 \\
Falco peregrinus & 248552 & Open & Fin & 1.029 & 3 & 14 & 0.990 & 5 & 5 & 0 & 6.498 & 5.881 & 0.951 & 2.000 & 15.00 & 19.00 \\
Lacerta agilis & 250100 & Pure & Bound & 0.996 & 4 & 15 & 1.000 & 6 & 6 & 0 & 6.578 & 6.578 & 1.000 & 2.000 & 15.00 & 22.00 \\
Rangifer tarandus & 249813 & Pure & Bound & 1.099 & 3 & 16 & 1.000 & 6 & 6 & 0 & 7.064 & 7.064 & 1.000 & 2.500 & 18.00 & 21.80 \\
Platalea minor & 240510 & Open & Fin & 1.107 & 4 & 30 & 1.000 & 6 & 6 & 0 & 7.110 & 7.096 & 0.999 & 3.000 & 22.00 & 26.00 \\
Pygoscelis adeliae & 240546 & Open & Fin & 1.002 & 2 & 7 & 0.871 & 5 & 4 & -1 & 7.180 & 3.995 & 0.653 & 4.000 & 18.00 & 21.00 \\
Branta leucopsis & 253148 & Open & Fin & 1.182 & 3 & 13 & 0.977 & 6 & 6 & 0 & 7.546 & 6.367 & 0.899 & 2.000 & 25.00 & 28.20 \\
Cervus elaphus & 249248 & Pure & Fin & 1.065 & 3 & 11 & 0.998 & 7 & 6 & -1 & 7.667 & 6.651 & 0.845 & 2.500 & 20.00 & 31.00 \\
Bostrychia hagedash & 248441 & Open & Fin & 0.935 & 3 & 12 & 0.981 & 7 & 6 & -1 & 7.669 & 6.349 & 0.874 & 3.000 & 12.00 & 16.00 \\
Parus major & 253153 & Open & Asym & 0.995 & 5 & 80 & 1.000 & 7 & 7 & 0 & 7.758 & 7.758 & 1.000 & 1.000 & 10.00 & 15.00 \\
Somateria mollissima & 248748 & Open & Fin & 0.901 & 2 & 7 & 0.863 & 6 & 4 & -2 & 7.907 & 4.209 & 0.597 & 3.000 & 30.00 & 37.80 \\
Cervus canadensis & 249243 & Pure & Fin & 1.017 & 2 & 8 & 0.909 & 7 & 5 & -2 & 8.092 & 5.071 & 0.599 & 2.500 & 18.00 & 25.00 \\
Connochaetes taurinus & 240670 & Pure & Fin & 0.969 & 3 & 11 & 0.961 & 7 & 6 & -1 & 8.129 & 6.489 & 0.776 & 2.500 & 18.00 & 24.00 \\
Odocoileus virginianus & 249516 & Open & Fin & 1.002 & 2 & 7 & 0.825 & 6 & 4 & -2 & 8.206 & 4.122 & 0.591 & 0.850 & 16.00 & 22.00 \\
Maccullochella peelii & 248132 & Open & Asym & 1.186 & 5 & 77 & 1.000 & 8 & 8 & 0 & 8.528 & 8.528 & 1.000 & 6.000 & 35.00 & 48.00 \\
Centrocercus minimus & 248466 & Open & Fin & 0.838 & 3 & 11 & 0.921 & 7 & 5 & -2 & 8.593 & 5.966 & 0.773 & 1.000 & 6.000 & 8.000 \\
Chrysemys picta & 250037 & Pure & Fin & 0.792 & 7 & 22 & 1.000 & 8 & 8 & 0 & 8.680 & 8.680 & 1.000 & 6.000 & 55.00 & 61.00 \\
Petronia petronia & 253154 & Open & Asym & 1.055 & 6 & 82 & 1.000 & 8 & 8 & 0 & 8.699 & 8.699 & 1.000 & 1.000 & 9.000 & 12.00 \\
Arctocephalus australis & 249126 & Open & Fin & 1.010 & 3 & 13 & 0.947 & 7 & 6 & -1 & 8.971 & 6.850 & 0.842 & 4.000 & 24.00 & 30.00 \\
Pernis apivorus & 248717 & Open & Fin & 0.968 & 4 & 15 & 0.973 & 7 & 7 & 0 & 9.040 & 7.497 & 0.886 & 3.000 & 20.00 & 27.00 \\
Phalacrocorax auritus & 248718 & Open & Fin & 0.999 & 4 & 15 & 0.969 & 7 & 7 & 0 & 9.174 & 7.538 & 0.880 & 2.500 & 18.00 & 24.00 \\
Rissa tridactyla & 253145 & Open & Fin & 1.120 & 5 & 25 & 0.999 & 8 & 8 & 0 & 9.177 & 8.958 & 0.990 & 4.000 & 24.00 & 28.00 \\
Macaca mulatta & 249412 & Open & Fin & 1.031 & 3 & 11 & 0.935 & 8 & 6 & -2 & 9.284 & 6.507 & 0.742 & 3.500 & 30.00 & 40.00 \\
Panthera pardus & 249672 & Open & Fin & 1.008 & 4 & 14 & 0.962 & 8 & 7 & -1 & 9.435 & 7.402 & 0.844 & 2.500 & 18.00 & 23.00 \\
Globicephala macrorhynchus & 240427 & Pure & Fin & 0.958 & 4 & 17 & 1.000 & 10 & 10 & 0 & 9.734 & 9.734 & 1.000 & 8.000 & 45.00 & 63.00 \\
Haliaeetus albicilla & 248635 & Pure & Fin & 1.040 & 5 & 20 & 0.992 & 8 & 8 & 0 & 9.754 & 8.926 & 0.948 & 5.000 & 35.00 & 42.00 \\
Larus Larus argentatus & 253146 & Open & Fin & 1.130 & 5 & 20 & 0.993 & 8 & 8 & 0 & 9.802 & 8.987 & 0.952 & 4.000 & 40.00 & 49.00 \\
Genypterus blacodes & 248123 & Open & Fin & 0.970 & 5 & 19 & 0.986 & 8 & 8 & 0 & 10.24 & 9.063 & 0.927 & 5.000 & 22.00 & 30.00 \\
Eumetopias jubatus & 249287 & Pure & Fin & 1.000 & 4 & 16 & 0.990 & 9 & 8 & -1 & 10.29 & 8.948 & 0.884 & 4.500 & 25.00 & 31.00 \\
Phocarctos hookeri & 249746 & Pure & Fin & 0.999 & 5 & 18 & 0.991 & 9 & 8 & -1 & 10.33 & 9.314 & 0.918 & 4.000 & 20.00 & 26.00 \\
Callorhinus ursinus & 249159 & Pure & Fin & 1.000 & 3 & 20 & 1.000 & 10 & 10 & 0 & 10.38 & 10.38 & 1.000 & 4.500 & 20.00 & 26.00 \\
Calyptorhynchus lathami & 248452 & Open & Fin & 1.034 & 3 & 10 & 0.841 & 8 & 6 & -2 & 10.52 & 5.856 & 0.632 & 4.000 & 40.00 & 50.00 \\
Ursus arctos & 249902 & Pure & Bound & 0.997 & 6 & 21 & 1.000 & 10 & 10 & 0 & 10.62 & 10.62 & 1.000 & 4.500 & 35.00 & 47.00 \\
Strix occidentalis & 248769 & Open & Fin & 0.935 & 2 & 7 & 0.703 & 9 & 4 & -5 & 11.66 & 4.406 & 0.437 & 2.000 & 22.00 & 25.00 \\
Fratercula arctica & 253143 & Open & Fin & 1.089 & 5 & 16 & 0.950 & 10 & 8 & -2 & 11.81 & 8.937 & 0.807 & 5.000 & 38.00 & 45.00 \\
Mirounga leonina & 249504 & Open & Fin & 0.983 & 6 & 22 & 0.987 & 10 & 9 & -1 & 11.90 & 10.63 & 0.930 & 4.000 & 18.00 & 23.00 \\
Gyps fulvus & 253144 & Open & Fin & 1.092 & 4 & 12 & 0.868 & 10 & 7 & -3 & 12.34 & 7.258 & 0.639 & 5.000 & 35.00 & 41.00 \\
Halichoerus grypus & 249376 & Open & Fin & 1.077 & 5 & 18 & 0.966 & 11 & 9 & -2 & 12.66 & 10.17 & 0.838 & 4.000 & 35.00 & 46.00 \\
Ursus americanus & 249888 & Open & Fin & 1.013 & 2 & 4 & 0.709 & 10 & 3 & -7 & 12.77 & 2.919 & 0.133 & 3.500 & 25.00 & 34.00 \\
Podocnemis lewyana & 250117 & Open & Fin & 0.912 & 4 & 12 & 0.853 & 10 & 7 & -3 & 12.78 & 7.292 & 0.621 & 10.00 & 30.00 & 35.00 \\
Papio cynocephalus & 249681 & Open & Fin & 1.008 & 7 & 29 & 0.996 & 12 & 12 & 0 & 13.80 & 13.06 & 0.968 & 4.500 & 30.00 & 37.00 \\
Coragyps atratus & 248510 & Open & Fin & 1.057 & 8 & 32 & 0.998 & 12 & 12 & 0 & 14.02 & 13.47 & 0.979 & 4.000 & 20.00 & 25.00 \\
Puffinus auricularis & 248731 & Open & Fin & 1.001 & 5 & 15 & 0.879 & 12 & 9 & -3 & 14.57 & 9.123 & 0.665 & 4.000 & 16.00 & 20.00 \\
Ovis canadensis & 249609 & Open & Fin & 0.956 & 3 & 11 & 0.792 & 11 & 7 & -4 & 14.75 & 7.221 & 0.502 & 2.500 & 16.00 & 24.00 \\
Theropithecus gelada & 252841 & Open & Fin & 1.022 & 6 & 18 & 0.914 & 12 & 10 & -2 & 15.24 & 10.60 & 0.737 & 4.500 & 30.00 & 37.00 \\
Anarhynchus frontalis & 248395 & Open & Fin & 0.947 & 3 & 9 & 0.603 & 14 & 6 & -8 & 18.58 & 5.752 & 0.353 & 2.000 & 24.00 & 30.00 \\
Chelydra serpentina & 250031 & Open & Fin & 1.007 & 12 & 60 & 1.000 & 17 & 17 & 0 & 18.92 & 18.86 & 0.999 & 8.000 & 40.00 & 47.00 \\
Cercopithecus mitis & 249242 & Open & Fin & 1.014 & 8 & 23 & 0.915 & 17 & 14 & -3 & 20.49 & 14.33 & 0.712 & 4.000 & 28.00 & 36.00 \\
Himantopus novaezelandiae & 248647 & Open & Fin & 1.004 & 3 & 9 & 0.565 & 15 & 6 & -9 & 20.49 & 5.778 & 0.322 & 2.000 & 15.00 & 20.00 \\
Phoebastria immutabilis & 248722 & Open & Fin & 0.991 & 9 & 24 & 0.918 & 17 & 14 & -3 & 21.15 & 14.86 & 0.718 & 8.000 & 55.00 & 68.00 \\
Propithecus verreauxi & 249768 & Open & Fin & 1.002 & 7 & 20 & 0.855 & 17 & 13 & -4 & 21.78 & 12.97 & 0.598 & 4.000 & 26.00 & 32.00 \\
Haematopus ostralegus & 248632 & Open & Fin & 0.983 & 2 & 6 & 0.437 & 16 & 4 & -12 & 22.46 & 3.905 & 0.212 & 4.000 & 35.00 & 43.00 \\
Fulmarus glacialis & 253142 & Open & Fin & 1.061 & 12 & 34 & 0.974 & 19 & 18 & -1 & 22.92 & 19.38 & 0.867 & 8.000 & 45.00 & 51.00 \\
Crocodylus johnsoni & 250050 & Pure & Fin & 1.034 & 13 & 54 & 1.000 & 24 & 24 & 0 & 24.71 & 24.71 & 1.000 & 12.00 & 60.00 & 70.00 \\
Clemmys guttata & 250557 & Pure & Fin & 0.930 & 8 & 21 & 0.777 & 27 & 15 & -12 & 25.90 & 15.08 & 0.350 & 10.00 & 55.00 & 65.00 \\
Elephas maximus & 249274 & Open & Fin & 1.053 & 16 & 41 & 0.994 & 25 & 24 & -1 & 28.12 & 25.04 & 0.873 & 12.00 & 65.00 & 79.00 \\
Gorilla beringei beringei & 249343 & Open & Fin & 1.000 & 10 & 26 & 0.846 & 23 & 17 & -6 & 28.81 & 17.23 & 0.583 & 8.000 & 35.00 & 40.00 \\
Brachyteles hypoxanthus & 249156 & Open & Fin & 1.031 & 9 & 23 & 0.777 & 24 & 15 & -9 & 30.59 & 15.58 & 0.487 & 7.000 & 32.00 & 40.00 \\
Pongo abelii & 249753 & Pure & Fin & 0.990 & 16 & 36 & 0.947 & 30 & 25 & -5 & 31.21 & 25.32 & 0.671 & 15.00 & 50.00 & 60.00 \\
Caretta caretta & 252676 & Open & Asym & 0.946 & 25 & 208 & 1.000 & 30 & 30 & 0 & 32.16 & 32.16 & 1.000 & 20.00 & 60.00 & 73.00 \\
Cebus capucinus & 249194 & Open & Fin & 1.001 & 7 & 18 & 0.598 & 29 & 12 & -17 & 38.31 & 12.58 & 0.303 & 6.000 & 40.00 & 54.00 \\
Pan troglodytes schweinfurthii & 249646 & Open & Fin & 1.000 & 16 & 36 & 0.764 & 39 & 25 & -14 & 49.65 & 25.32 & 0.455 & 11.00 & 45.00 & 60.00 \\
Huso huso & 252953 & Pure & Fin & 1.000 & 14 & 33 & 0.784 & 50 & 27 & -23 & 52.72 & 26.50 & 0.221 & 16.00 & 90.00 & 118.0 \\
\hline
\end{tabular}
}
\end{table}

\FloatBarrier

\bibliographystyle{abbrv}
\bibliography{BibloH2}
\end{document}